\documentclass{article}

\usepackage{arxiv}
\usepackage{float}
\usepackage[utf8]{inputenc}
\usepackage{hyperref}
\hypersetup{
	colorlinks=true, 
	citecolor=blue,
	linkcolor=black,
	urlcolor=blue
}
\usepackage{setspace}
\usepackage{amsmath,amssymb,amsthm}
\usepackage{graphicx}
\usepackage{enumerate}
\usepackage{algorithm}
\usepackage{algpseudocode}
\usepackage{tikz}
\usetikzlibrary{arrows, calc, backgrounds, shapes.misc, positioning, shapes, fit}

\newcommand{\figref}[1]{Figure~\ref{#1}}

\newcommand{\bmat}[1]{\begin{bmatrix}#1\end{bmatrix}}
\newcommand{\bmtx}{\begin{bmatrix}}
	\newcommand{\emtx}{\end{bmatrix}}
\newcommand{\bsmtx}{\left[ \begin{smallmatrix}} 
	\newcommand{\esmtx}{\end{smallmatrix} \right]}
\newcommand{\field}[1]{\mathbb{#1}}
\newcommand{\R}{\field{R}}
\newcommand{\Sm}{\field{S}}
\newcommand{\Ltwo}{\mathcal{L}_2[0,T]}
\newcommand{\RH}{\field{RH}_\infty}
\newcommand{\RL}{\field{RL}_\infty}
\makeatletter

\newcommand{\Rmnum}[1]{\expandafter\@slowromancap\romannumeral #1@}
\makeatother
\newcommand{\commentr}[1]{\textcolor{red}{#1}}
\newcommand{\commentb}[1]{\textcolor{blue}{#1}}

\definecolor{mycolor1}{HTML}{367D7D}
\definecolor{mycolor2}{HTML}{D33502}
\definecolor{mycolor3}{HTML}{FAA818}
\definecolor{mycolor4}{HTML}{41A30D}
\definecolor{mycolor5}{HTML}{FFCE38}
\definecolor{mycolor6}{HTML}{6EBCBC}
\definecolor{mycolor7}{HTML}{37526D}

\colorlet{plant}{mycolor5!40!white}
\colorlet{unc}{mycolor6!40!white}
\colorlet{closedloop}{gray!20!white}
\colorlet{controller}{mycolor4!30!white}
\colorlet{dist}{mycolor3!60!white}
\colorlet{marker}{mycolor7}
\colorlet{algo}{mycolor3!60!white}
\colorlet{miscplot}{mycolor2}
\colorlet{clplot}{mycolor2}
\colorlet{filter}{magenta!20!white}

\newtheorem{theorem}{Theorem}
\newtheorem{definition}{Definition}
\newtheorem{ex}{Example}
\newtheorem{lem}{Lemma}
\newtheorem{asm}{Assumption}
\newtheorem{eremark}{Remark}

\tikzset{>=stealth'}

\title{Finite Horizon Robust Synthesis\\Using Integral Quadratic Constraints}

\author{
  Jyot Buch\\
  Department of Aerospace Engineering and Mechanics\\
  University of Minnesota, Minneapolis\\
  Minneapolis, MN 55455 \\
  \texttt{buch0271@umn.edu} \\
   \And
 Peter Seiler\\
  Department of Electrical Engineering\\
  University of Michigan, Ann Arbor\\
  Ann Arbor, MI 48109 \\
  \texttt{pseiler@umich.edu} \\
}

\begin{document}
\maketitle

\begin{abstract}
	We present a robust synthesis algorithm for uncertain
	linear time-varying (LTV) systems on finite horizons. The uncertain
	system is described as an interconnection of a known LTV system and
	a perturbation. The input-output behavior of the perturbation is
	specified by time-domain Integral Quadratic Constraints (IQCs). The
	objective is to synthesize a controller to minimize the worst-case
	performance. This leads to a non-convex optimization. The proposed
	approach alternates between an LTV synthesis step and an IQC
	analysis step. Both induced $\mathcal{L}_2$ and terminal Euclidean 
	norm penalties on output are considered for finite horizon performance.
	The proposed algorithm ensures that the robust performance is non-increasing at each iteration step. The effectiveness of this method is demonstrated using numerical examples.
\end{abstract}


\section{Introduction}

This paper considers robust synthesis for uncertain linear
time-varying (LTV) systems on finite horizons. This problem is motivated by engineering systems that follow a finite-time trajectory and for which model uncertainty is a significant factor.  Examples of such systems include: aircraft landings~\cite{subrahmanyam2012finite}, missile interceptors~\cite{tucker1998continuous} and space-launch or reentry systems~\cite{marcos09,biertumpfel2019finite,man09}. The Jacobian linearization of the nonlinear dynamics along the trajectory yields an uncertain, finite horizon LTV system.  Robust synthesis can be used to ensure the stability and robustness of the linearized closed-loop over a range of parametric and dynamic uncertainties. Many existing robust synthesis algorithms, e.g. $\mu$-synthesis
\cite{doyle1987design,balas1992design,packard1993linear,veenman2014iqc} have been developed for
uncertain linear time-invariant (LTI) system and infinite horizon
robustness metrics. This enables the use of frequency-domain
techniques. In contrast, this paper is developed for uncertain
finite horizon, LTV systems using time-domain techniques.

The specific formulation uses an uncertain system
described by an interconnection of a known LTV system and a
perturbation. The input-output behavior of the perturbation is
described by time-domain Integral Quadratic Constraints (IQCs). 
The performance objective is specified by an induced gain from
$\mathcal{L}_2$ input disturbances to a mixture of an $\mathcal{L}_2$
and terminal Euclidean norm on the output. The objective is to
synthesize a controller to minimize the worst-case performance over
all allowable uncertainties. This worst-case performance can be used
to robustly bound the state at the end of a finite horizon in the
presence of external disturbances and model uncertainty.

This robust synthesis problem leads, in general, to a non-convex
optimization. The proposed algorithm, presented in
Section~\ref{sec:RobustSynthesis}, iterates between a nominal
synthesis step and robustness analysis step. The nominal synthesis
step relies on existing finite horizon $H_\infty$
synthesis results which consider a control theoretic formulation~\cite{subrahmanyam2012finite,petersen00,green2012linear,tadmor1990worst}. An alternative game theoretic formulation is considered in~\cite{limebeer1992game} which provides equivalent synthesis conditions. These conditions can be stated in terms of two coupled Riccati Differential Equations (RDEs)~\cite{khargonekar1991h_,ravi1991h,uchida1992finite} or two coupled Riccati Differential Inequalities (RDIs)~\cite{lall1995riccati}. We use the two coupled RDEs as it provides numerical advantage over the RDI conditions. Moreover, in contrast to other work, the results in~\cite{khargonekar1991h_,ravi1991h,uchida1992finite} allow for terminal Euclidean norm penalties on the output. The robustness analysis step uses the IQC framework introduced in~\cite{yakubovich1967frequency, megretski1997system}. This framework has been extended in~\cite{seiler2019finite} to assess robustness of the uncertain LTV systems on finite horizons. The approach presented in~\cite{seiler2019finite} will be used in this paper for the robustness analysis. Finally, a scaled plant construction is required to link the nominal synthesis and robustness analysis steps.

The proposed method is analogous to the existing \mbox{DK} iteration method
for uncertain LTI systems on infinite horizons.  The algorithm in
this paper generalizes this method to uncertain LTV systems on finite
horizons. Similar extensions have been made in~\cite{wang2016robust,veenman2010robust}
for Linear Parameter-Varying (LPV) systems. Two other closely related works are~\cite{o1999robust} and \cite{pirie2002robust}. The work in~\cite{o1999robust} considers an extension of the Glover-McFarlane
loop-shaping method to LTV systems on infinite horizons. This leads to a robust
stabilization problem with a single full block uncertainty. The work in~\cite{pirie2002robust} provides convex synthesis conditions for robust
performance of uncertain LTV systems. However, \cite{pirie2002robust} assumes
that uncertainty lies in a contractive subset and is block partitioned with $(2,2)$ block being zero. This special structure is used to convexify the synthesis optimization. The algorithm proposed in this paper considers more general robust performance formulation than in~\cite{o1999robust} and~\cite{pirie2002robust}, which allows us to design output-feedback controllers that robustly bound the reachable set of a finite horizon LTV system. A MATLAB implementation of the proposed algorithm including the numerical examples are available in the LTVTools~\cite{LTVTools} toolbox.

There are three main contributions of the paper. We propose a new iterative algorithm to synthesize robust output feedback controllers of uncertain LTV systems on finite time horizons. This is a continuation of our preliminary work in~\cite{buch2020finte}. The distinctions from~\cite{buch2020finte} are as follows: First, we use the dynamic IQC multipliers for the proposed algorithm, whereas the prior work in~\cite{buch2020finte} used the memoryless IQCs and related classes of uncertainties. Second, we use a time-varying IQC factorization to construct a scaled plant. This step ensures that the worst-case gain at each iteration is monotonically non-increasing. Finally, this paper provides all details and technical proofs regarding the proposed approach. The effectiveness is demonstrated using a nonlinear robot arm example.

\vspace{0.1in}
\noindent\textbf{Notation:} Let $\mathbb{R}^{n \times m}$ and
$\Sm^{n}$ denote the sets of $n$-by-$m$ real matrices and $n$-by-$n$
real, symmetric matrices. The finite horizon $\Ltwo$ norm of a (Lebesgue integrable) signal $v:[0,T] \rightarrow \R^n$ is
$\|v\|_{2, [0,T]} := \left( \int_0^T v(t)^\top v(t) dt \right)^{1/2}$. If
$\|v\|_{2, [0,T]} < \infty$ then $v \in \mathcal{L}_2^n[0,T]$. $\RL$ is the set
of rational functions with real coefficients that are proper and have no poles on the imaginary axis.  $\RH \subset \RL$ contains functions that are
analytic in the closed right-half of the complex plane. 
An abstract formulation using standard Linear Fractional Transformation (LFT) framework~\cite{zhou1996robust,dullerud2013course} is used throughout the paper. The notations $\mathcal{F}_{l}(G,K)$ and $\mathcal{F}_{u}(N,\Delta)$ represents lower and upper LFTs respectively. Finally, $G^\sim$ denotes the adjoint of a dynamical system $G$ as formally defined in Section 3.2.4 of \cite{green2012linear}.

\section{Preliminaries}
\label{sec:prelim}

\subsection{Nominal Performance}
\label{sec:nomperf}
Consider an LTV system $H$ defined on the horizon $[0,T]$:
\begin{align}
	\label{eq:nLTV}
	\dot{x}(t) &= A(t)\, x(t) + B(t)\, d(t)\\
	e(t) & = C(t) \, x(t) + D(t) \, d(t)
\end{align}
where $x(t) \in \R^{n_x}$ is the state, $d(t) \in \R^{n_d}$ is the disturbance input, and
$e(t) \in \R^{n_e}$ is the performance output at time $t\in[0,T]$.  The state matrices $A : [0,T] \rightarrow \R^{n_x \times n_x}$.
$B : [0,T] \rightarrow \R^{n_x \times n_d}$,
$C : [0,T] \rightarrow \R^{n_e \times n_x}$, and
$D : [0,T] \rightarrow \R^{n_e \times n_d}$ are piecewise-continuous
(bounded) real matrix valued functions of time.
It is assumed throughout
that $T < \infty$. Thus $d \in \mathcal{L}_2[0,T]$ implies $x$ and $e$
are in $\mathcal{L}_2[0,T]$ for any initial condition $x(0)$
(Chapter 3 of \cite{green2012linear}). To simplify further, zero initial conditions are assumed for states i.e.  $x(0) = 0$. Explicit time dependence of the state
matrices is omitted when it is clear from the context. The performance of $H$ will be assessed in terms of an induced gain
with two components. First partition the output as follows:
\begin{align}
	\label{eq:LTVoutput}
	\bmtx e_I(t) \\ e_E(t) \emtx =
	\bmtx C_I(t) \\ C_E(t) \emtx \, x(t) 
	+ \bmtx D_I(t) \\ 0 \emtx \, d(t) 
\end{align}
where $e_I(t) \in \R^{n_I}$ and $e_E(t) \in \R^{n_E}$ with $n_e = n_E+n_I$.
The generalized performance metric of $H$ is then defined as,
\begin{align}
	\|H\|_{[0,T]}:=
	\sup_{\stackrel{0\ne d \in \mathcal{L}_2[0,T]}{x(0)=0}}
	\left[
	\frac{{\|e_E(T)\|_2^2 + \|e_I\|^2_{2,[0,T]}}} {\|d\|^2_{2,[0,T]}} 
	\right]^{1/2}
	\label{eq:norm}
\end{align}
This defines an induced gain from the input $d$ to a mixture of an
$\mathcal{L}_2$ and terminal Euclidean norm on the output $e$. This is a useful generalization, as many control design requirements often involve bounding the outputs at final time in addition to bounded control effort. The example discussed in Section~\ref{sec:NLEx} uses such mixed penalties. More general quadratic cost as in~\cite{seiler2019finite} can also be considered with appropriate choice of the input-output matrices (see Appendix~\ref{sec:quadcost}). Note that if $n_E=0$ then there is no terminal Euclidean norm penalty on the output. This case corresponds to the standard, finite horizon induced $\mathcal{L}_2$ gain of $H$. Similarly, if $n_I=0$ then there
is no $\mathcal{L}_2$ penalty on the output. This case corresponds to
a finite horizon $\mathcal{L}_2$-to-Euclidean gain. This can be used
to bound the terminal output $e_E(T)$ resulting from an
$\mathcal{L}_2$ disturbance input. Zero feed-through from $d$ to $e_E$
ensures that the Euclidean penalty is well-defined at time $t=T$. The next theorem states an equivalence between a bound on this
performance metric $\|H\|_{[0,T]}$ and the existence of a solution to
a related RDE (Theorem 3.7.4 of \cite{green2012linear}).
\begin{theorem}
	\label{thm:nominalperf}
	Consider an LTV system \eqref{eq:nLTV} with $\gamma>0$ given. Let
	$Q : [0,T] \rightarrow \Sm^{n_x}$,
	$S : [0,T] \rightarrow \R^{n_x \times n_d}$,
	$R : [0,T] \rightarrow \Sm^{n_d}$, and $F \in \R^{n_x\times n_x}$ be
	defined as follows\footnote{If $n_I=0$ then $Q=0_{n_x}$, $S=0_{n_x \times n_d}$, and $R=-\gamma^2I_{n_d}$. Similarly, if $n_E=0$ then  $F=0_{n_x}$.}.
	\begin{align*}
		\begin{split}
			Q := C_I^\top C_I, \hspace{0.1in}  
			S := C_I^\top D_I, \hspace{0.1in}
			R := D_I^\top D_I-\gamma^2 I_{n_d}, \hspace{0.1in}
			F := C_{E}(T)^\top C_{E}(T)
		\end{split}
	\end{align*}
	The following statements are equivalent:
	\begin{enumerate}
		\item $\|H\|_{[0,T]} < \gamma$
		\item $R(t) \prec 0$ for all $t \in [0,T]$. Moreover, there
		exists a differentiable function $P:[0,T] \rightarrow \Sm^{n_x}$ such
		that $P(T)=F$ and
		\begin{align*}
			\dot{P} +  A^\top P + PA + Q- (PB+S) R^{-1} (PB+S)^\top = 0
		\end{align*}
		This is a Riccati Differential Equation (RDE).
	\end{enumerate}  
\end{theorem}

The nominal performance $\|H\|_{[0,T]} < \gamma$ is achieved if the
associated RDE solution exists on $[0,T]$ when integrated backward
from $P(T)=F$. The assumption $R(t)\prec 0$ ensures $R(t)$ is
invertible and hence the RDE is well-defined $\forall t\in [0,T]$.
Thus, the solution of the RDE exists on $[0,T]$ unless it grows
unbounded. The smallest bound on $\gamma$ is obtained using
bisection.

\subsection{Nominal Synthesis}
\label{sec:nominalsyn}

This subsection provides conditions to synthesize a controller that is
optimal with respect to the nominal performance metric introduced in
the previous subsection.  Consider the feedback interconnection shown
in \figref{fig:NominalOpFb}.
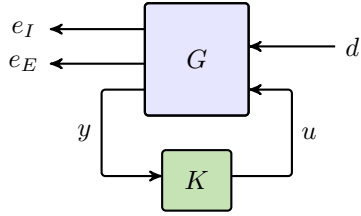
\begin{figure}[h]
	\centering
	\begin{tikzpicture}[thick,rounded corners = 0.5mm,scale=1.15]
		\draw [->](0,0.9) -- (-1.1,0.9);
		\draw [->](0,0.5) -- (-1.1,0.5);
		\draw [->](2.2,0.7) -- (1.2,0.7);
		\node at (2.4,0.7) {$d$};
		\node at (-1.4,0.9) {$e_I$};
		\node at (-1.4,0.5) {$e_{E}$};
		\draw [fill=controller]  (0.2,-0.5) rectangle node{$K$}(1,-1.2);
		\draw [->](0,0.2) -- (-0.5,0.2) -- (-0.5,-0.8) -- (0.2,-0.8);
		\draw [->](1,-0.8) -- (1.7,-0.8) -- (1.7,0.2) -- (1.2,0.2);
		\node at (-0.7,-0.3) {$y$};
		\node at (1.9,-0.3) {$u$};
		\draw [fill=blue,fill opacity=0.1] (0,-0.1) rectangle node{$G$}(1.2,1.2);
		\draw (0,-0.1) rectangle node{$G$}(1.2,1.2);
	\end{tikzpicture}	
	\caption{Nominal Feedback Interconnection $\mathcal{F}_l(G,K)$}
	\label{fig:NominalOpFb}
\end{figure}

\noindent The LTV system $G$ defined on $[0,T]$ is given by:
\begin{align}
	\label{eq:nomsyn}
	\bmtx \dot{x}(t)\\\hline e_{I}(t) \\ e_E(t) \\y(t)\emtx =
	\left[
	\begin{array}{c|c c}
		A(t) & B_d(t) & B_u(t)\\
		\hline
		C_{I}(t) & 0 & D_{Iu}(t) \\
		C_E(t) & 0 & 0\\
		C_y(t) & D_{yd}(t) & 0
	\end{array}
	\right]
	\bmtx x(t) \\\hline d(t)\\u(t)\emtx
\end{align}
where $d(t) \in \R^{n_d}$ is the generalized disturbance,
$u(t) \in \R^{n_u}$ is the control input and $y(t) \in \R^{n_y}$ is the
measured output. The generalized disturbance is of the form $d(t) = \bsmtx d_{in}(t)\\ n(t)\esmtx$, where $n(t) \in \R^{n_y}$ is a measurement noise and $d_{in}(t)$ represents all other disturbance inputs. This plant structure also assumes no feedthrough from $d$ to $e_E$. 
This is required to ensure that the nominal performance metric is well-posed.
In addition, the standard $H_\infty$ synthesis framework imposes 
additional structure on the matrices relating $d$ to $e_I$ and $d$ to $y$. This is done to 
simplify notation and is obtained via standard loop
transformations under some minor technical assumptions
(Chapter 17 of \cite{zhou1996robust}). This leads to the following additional structure 
on the plant matrices:
\begin{align*}
	C_I:= \bmtx 0 \\ C_1\emtx \;\;\;\;
	D_{Iu} := \bmtx I_{n_u}\\ 0\emtx \;\;\;\;\;
	D_{yd} := \bmtx 0 & I_{n_y}\emtx
\end{align*}
The nominal synthesis problem is to find a causal linear time-varying
controller
$K:\mathcal{L}^{n_y}_2[0,T]\rightarrow\mathcal{L}^{n_u}_2[0,T]$ that
optimizes the closed-loop nominal performance, i.e.:
\begin{align*}
	\inf_{K} \, \|\mathcal{F}_l(G,K)\|_{[0,T]}
\end{align*}
As noted previously, if $n_E=0$ then the nominal performance metric is
the (finite horizon) induced $\mathcal{L}_2$ gain.  In this case, the
synthesis problem is equivalent to the existing finite horizon
$H_\infty$ problem as considered in
\cite{green2012linear,tadmor1990worst}. The theorem below states the
necessary and sufficient conditions for existence of a
$\gamma$-suboptimal controller for the nominal performance metric
(with $n_E$ not necessarily equal to zero). Theorem \ref{thm:nomsyn} is a
special case of results presented in
\cite{khargonekar1991h_,uchida1992finite}.

\begin{theorem} 
	\label{thm:nomsyn}	 
	Consider an LTV system \eqref{eq:nomsyn} with $\gamma > 0$ given. Let $B$, $\hat{C}$, $\bar{R}$ and $\hat{R}$ be defined as follows.
	\begin{align*}
		\begin{split}
			B:= \bmtx B_d & B_u\emtx, \hspace{0.1in}   \bar{R}:=\text{diag}\{-\gamma^{2}I_{n_d}, I_{n_u}\}, \hspace{0.1in}  
			\hat{C}:= \bmtx C_I^\top & C_y^\top\emtx^\top, \hspace{0.1in}   \hat{R}:=\text{diag}\{-\gamma^{2}I_{n_I}, I_{n_y}\}
		\end{split}
	\end{align*}
	\begin{enumerate}		
		\item There exists an admissible output feedback controller $K$ such that $\|\mathcal{F}_l(G,K)\|_{[0,T]}<\gamma$ if and only if the following three conditions hold:
		\begin{enumerate}			
			\item There exists a differentiable function
			$X:[0,T] \rightarrow \Sm^{n_x}$ such that $X(T) = C_E(T)^\top C_E(T)$,		
			\begin{align*}
				\dot{X} + A^{\top} X + XA - XB\bar{R}^{-1}B^\top X + C_{I}^\top C_{I}= 0 
			\end{align*}
			\item There exists a differentiable function
			$Y:[0,T] \rightarrow \Sm^{n_x}$ such that $Y(0) = 0$,		
			\begin{align*}
				-\dot{Y} + AY + YA^{\top} - Y\hat{C}^\top\hat{R}^{-1}\hat{C}Y + B_{d}B_{d}^\top= 0
			\end{align*}
			\item $X(t)$ and $Y(t)$ satisfy the following point-wise in time spectral radius condition,
			\begin{align}
				\label{ccond}
				\rho(X(t)Y(t))<\gamma^2, \hspace{0.1in} \forall t \in [0,T]
			\end{align}
		\end{enumerate}
		\item If the conditions above are satisfied, then the closed loop performance
		$\|\mathcal{F}_l(G,K)\|_{[0,T]}<\gamma$ is achieved by the
		following central controller:
		\begin{align*}
			\dot{\hat{x}}(t) &= A_K(t)\, \hat{x}(t) + B_K(t)\, y(t)\\
			u(t) &= C_K(t)\, \hat{x}(t)
		\end{align*}
		where
		\begin{align*}
			Z &  := (I - \gamma^{-2}YX)^{-1} \\
			A_K& := A + \gamma^{-2}B_dB_d^\top X - ZYC_y^\top C_y - B_uB_u^\top X\\
			B_K& := ZYC_y^\top\\
			C_K& :=-B_u^\top X 
		\end{align*}
	\end{enumerate}
\end{theorem}
For a given $\gamma>0$, the RDEs associated with $X$ and $Y$
are integrated backward and forward in time, respectively.  If
solution to both RDEs exist then the spectral radius coupling
condition \eqref{ccond} is checked. If all three conditions are
satisfied then the central controller achieves a closed-loop
performance of $\gamma$.  The smallest possible value of $\gamma$ is
obtained using bisection. The results in
\cite{khargonekar1991h_,uchida1992finite} also consider the effect of
uncertain initial conditions.

\section{Robust Performance}
\label{sec:robperf}

\subsection{Uncertain LTV Systems}
An uncertain, time-varying system $\mathcal{F}_u(N,\Delta)$ is shown
in \figref{fig:UncertainOL}. This consists of an interconnection of a
known finite horizon LTV system $N$ and a perturbation $\Delta$. This
perturbation represents block-structured uncertainties and/or
nonlinearities. 
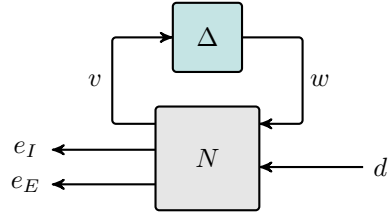
\begin{figure}[h]
	\centering
	\begin{tikzpicture}[thick,scale=1.15,rounded corners = 0.5mm]
		\draw [fill=closedloop] (0,0) rectangle node{$N$}(1.2,1.2);
		\draw [fill=unc] (0.2,1.6) rectangle node{$\Delta$}(1,2.4);
		\draw [->](0,1) -- (-0.5,1) -- (-0.5,2) -- (0.2,2);
		\draw [->](1,2) -- (1.7,2) -- (1.7,1) -- (1.2,1);
		\draw [->](0,0.7) -- (-1.2,0.7);
		\draw [->](0,0.3) -- (-1.2,0.3);
		\draw [->](2.4,0.5) -- (1.2,0.5);
		\node at (2.6,0.5) {$d$};
		\node at (1.9,1.5) {$w$};
		\node at (-0.7,1.5) {$v$};
		\node at (-1.5,0.7) {$e_I$};
		\node at (-1.5,0.3) {$e_E$};
	\end{tikzpicture}	
	\caption{Uncertain System Interconnection $\mathcal{F}_u(N,\Delta)$}
	\label{fig:UncertainOL}
\end{figure}
The term ``uncertainty'' is used for simplicity when
referring to $\Delta$. It is assumed throughout that the interconnection
$\mathcal{F}_u(N,\Delta)$ is well-posed. A formal definition for
well-posedness is given in
\cite{zhou1996robust,megretski1997system}. The LTV system $N$ is
described by the following state-space model:
\begin{align}
	\label{eq:uLTV}
	\bmtx \dot{x}_N(t)\\\hline v(t)\\ e_{I}(t) \\ e_E(t)\emtx =
	\left[
	\begin{array}{c|c c}
		A_N(t) & B_w(t) & B_d(t)\\
		\hline
		C_v(t) & D_{vw}(t) & D_{vd}(t) \\
		C_{I}(t) & D_{Iw}(t) & D_{Id}(t)\\
		C_E(t) & 0 & 0
	\end{array}
	\right]
	\bmtx x_N(t) \\\hline w(t)\\ d(t)\emtx
\end{align}
In addition to notations defined earlier $v \in \R^{n_v}$ and
$w \in \R^{n_w}$ are signals associated with the uncertainty $\Delta$. The state vector is denoted as $x_N \in \R^{n_N}$ to refer to the states of system $N$.  

\subsection{Worst-Case Gain}
The robust performance of the uncertain system $\mathcal{F}_u(N,\Delta)$ is assessed using the worst-case gain as defined below.
\begin{definition}
	\label{def:wcgain}
	Let an LTV system N be given by \eqref{eq:uLTV} and uncertainty	$\Delta:\mathcal{L}^{n_v}_2[0,T]\rightarrow\mathcal{L}^{n_w}_2[0,T]$
	be in some set $\mathcal{S}$. Assume the interconnection $\mathcal{F}_u(N,\Delta)$ is well-posed. The worst-case gain is then defined as:
	\begin{align*}
		\gamma_{wc} := \sup_{\Delta \in \mathcal{S}} \|\mathcal{F}_u(N,\Delta)\|_{[0,T]}
	\end{align*}
\end{definition}
The worst-case gain is the largest induced gain of the uncertain time-varying system over all uncertainties $\Delta$ in set $\mathcal{S}$. This is difficult to compute directly as it involves an optimization over the entire uncertainty set. Instead, we focus on computing an upper bound on the worst-case gain using dissipation inequalities and IQC conditions.

\subsection{Integral Quadratic Constraints (IQCs)}
IQCs \cite{yakubovich1967frequency,megretski1997system} are used to describe the input-output
behavior of $\Delta$. A time-domain formulation is used here for the analysis of the uncertain time-varying system. This formulation is based on the graphical interpretation as shown in \figref{fig:IQCFilter}. Time domain IQCs, as used in this paper, are defined for $\Delta$ by specifying a filter $\Psi$ and a finite horizon constraint on the filter output $z$.
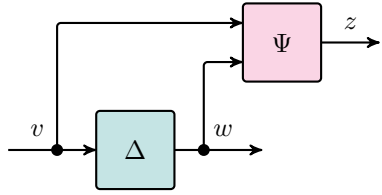
\begin{figure}[h]
	\centering
	\begin{tikzpicture}[thick,scale=1.3,rounded corners = 0.5mm]
		\draw [fill=unc] (0.1,0) rectangle node{$\Delta$}(0.9,0.8);
		\draw [fill=filter] (1.6,1.1) rectangle node{$\Psi$}(2.4,1.9);
		\draw [->](-0.3,0.4) -- (-0.3,1.7) -- (1.6,1.7);
		\draw [->](1.2,0.4) -- (1.2,1.3) -- (1.6,1.3);
		\draw [->](0.9,0.4) -- (1.8,0.4);
		\draw [->](-0.8,0.4) -- (0.1,0.4);
		\draw [->](2.4,1.5) -- (3,1.5);
		\node at (1.4,0.6) {$w$};
		\node at (-0.5,0.6) {$v$};
		\node at (2.7,1.7) {$z$};
		\draw [fill=black] (-0.3,0.4) ellipse (0.05 and 0.05);
		\draw [fill=black] (1.2,0.4) ellipse (0.05 and 0.05);
	\end{tikzpicture}	
	\caption{Graphical Interpretation for Time Domain IQCs}
	\label{fig:IQCFilter}
\end{figure}

\noindent The LTV dynamics of filter $\Psi$ on the horizon $[0,T]$ are given as follows:
\begin{align}
	\label{eq:Psi}
	\bmtx \dot{x}_\psi(t) \\\hline z(t)\emtx =
	\left[
	\begin{array}{c|c c}
		A_\psi(t) \, & B_{\psi v}(t) \, & B_{\psi w}(t)\\
		\hline
		C_\psi(t) \, & D_{\psi v}(t) \, & D_{\psi w}(t)
	\end{array}
	\right]
	\bmtx x_\psi(t) \\\hline v(t) \\ w(t)\emtx
\end{align}
where $x_\psi \in \R^{n_\psi}$ is the state. The formal definition for a time-domain IQC is given next.
\begin{definition}
	\label{def:tdiqc}
	Consider an LTV system $\Psi:\mathcal{L}_{2}^{(n_v+n_w)}[0,T] \rightarrow \mathcal{L}_{2}^{n_z}[0,T]$ and
	$M:[0,T] \rightarrow \Sm^{n_z}$ be given with $M$ piecewise
	continuous. A bounded, causal operator
	$\Delta:\mathcal{L}_{2}^{n_v}[0,T] \rightarrow \mathcal{L}_{2}^{n_w}[0,T]$
	satisfies the \underline{time domain IQC} defined by $(\Psi, M)$ if
	the following inequality holds for all
	$v \in \mathcal{L}_2^{n_v}[0,T]$ and $w=\Delta(v)$:
	\begin{align}
		\label{eq:tdiqc}
		\int_0^T z(t)^\top M(t) z(t) \, dt \, \ge 0
	\end{align}
	where $z$ is the output of $\Psi$ driven by inputs $(v,w)$ with
	zero initial conditions $x_\psi(0)=0$.
\end{definition}
Note that Definition~\ref{def:tdiqc} allows the IQC filter $\Psi$ to be time-varying. This time-varying generalization provides an additional degree of freedom for finite horizon robustness analysis with IQCs. Similar generalizations for LPV systems are presented in~\cite{pfifer2015robustness} to use parameter-varying IQCs. However, exploring this additional degree of freedom is a subject of future research. Thus, the examples discussed later in the paper are for the special case where $\Psi$ is an LTI filter. The notation $\Delta \in \mathcal{I}(\Psi,M)$ is used if $\Delta$ satisfies the IQC defined by $(\Psi,M)$. A valid IQC $\mathcal{I}(\Psi,M)$ can be defined for a set $\mathcal{S}$ such that $\mathcal{S} \subseteq \mathcal{I}(\Psi,M)$. Two examples are provided below.
\begin{ex}
	\label{ex:LTIunc}
	Let $\mathcal{S}$ denote the set of LTI uncertainties $\Delta\in\RH$ with $\| \Delta \|_\infty \le 1$. Let $(\Psi,M)$ be defined as follows:
	\begin{align}
		\label{eq:LTIuncIQC}
		\begin{split}
			\Psi & :=\bmtx \Psi_{11} & 0 \\ 0 & \Psi_{11} \emtx 
			\mbox{ with } \Psi \in \RH^{n_z \times 1} \\
			M & :=\bmtx M_{11} & 0 \\ 0 & -M_{11} \emtx
			\mbox{ with } M\in \Sm^{n_z} \mbox{ and } M_{11} \succ 0
		\end{split}
	\end{align}
	It is shown in Appendix~\Rmnum{2} of \cite{balakrishnan2002lyapunov} that the pair $(\Psi,M)$ defines a valid time domain IQC for $\Delta$ over any $T<\infty$ i.e. $\mathcal{S} \subseteq \mathcal{I}(\Psi,M)$.
\end{ex}
\begin{ex}
	\label{ex:TVUnc}
	Let $\mathcal{S}$ be the set of LTV parametric
	uncertainties $\delta(t)\in \R$ with a given norm-bound $\beta(t)$, i.e. $w(t)=\delta(t) \cdot v(t)$, 
	$|\delta(t)| \le \beta(t)$, $\forall t\in[0,T]$. Let $n_v = n_w = n$ and 
	$M_{11}:[0,T]\rightarrow \Sm^{n}$ be piecewise continuous with
	$M_{11}(t) \succ 0$, $\forall t\in [0,T]$.  Then $\Delta$ satisfies the
	IQC defined by the time-varying matrix:
	\begin{align}
		M(t) := \bmtx \beta(t)^2 \, M_{11}(t) & 0\\0 & -M_{11}(t) \emtx
	\end{align}
	and a static filter $\Psi := I_{2n}$, i.e. $\mathcal{S} \subseteq \mathcal{I}(\Psi,M)$.
\end{ex}
A library of IQCs is provided in~\cite{megretski1997system,veenman2016robust} for
various types of perturbations. Most IQCs are for bounded, causal
operators with multipliers $\Pi$ specified in the frequency domain.
Under mild assumptions, a valid time-domain IQC $\mathcal{I}(\Psi,M)$ can be constructed from $\Pi$ via a $J$-spectral factorization
\cite{seiler2014stability}.

\subsection{Dissipation Inequality Condition}
Consider an extended system as shown in \figref{fig:IQCAugment}. This interconnection includes the IQC filter $\Psi$ but the uncertainty $\Delta$ has been removed. The
precise relation $w=\Delta(v)$ is replaced, for the analysis, by the
constraint on the filter output $z$.
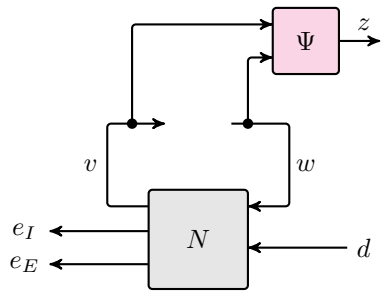
\begin{figure}[h]
	\centering
	\begin{tikzpicture}[thick,scale=1.1,rounded corners = 0.5mm]
		\draw [fill=closedloop] (0,0) rectangle node{$N$}(1.2,1.2);
		\draw [->](0,1) -- (-0.5,1) -- (-0.5,2) -- (0.2,2);
		\draw [->](1,2) -- (1.7,2) -- (1.7,1) -- (1.2,1);
		\draw [->](0,0.7) -- (-1.2,0.7);
		\draw [->](0,0.3) -- (-1.2,0.3);
		\draw [->](2.4,0.5) -- (1.2,0.5);
		\draw [->](-0.2,2) -- (-0.2,3.2)-- (1.5,3.2) ;
		\draw [->](1.2,2) -- (1.2,2.8)-- (1.5,2.8) ;
		\draw [->](2.3,3) -- (2.8,3);
		\node at (2.6,0.5) {$d$};
		\node at (1.9,1.5) {$w$};
		\node at (-0.7,1.5) {$v$};
		\node at (-1.5,0.7) {$e_I$};
		\node at (-1.5,0.3) {$e_E$};
		\draw [fill=black] (-0.2,2) ellipse (0.05 and 0.05);
		\draw [fill=black] (1.2,2) ellipse (0.05 and 0.05);
		\draw [fill=filter] (1.5,2.6) rectangle node{$\Psi$}(2.3,3.4);
		\node at (2.6,3.2) {$z$};
	\end{tikzpicture}	
	\caption{Analysis Interconnection}
	\label{fig:IQCAugment}
\end{figure}

\noindent The extended system of $N$ (Equation~\ref{eq:uLTV}) and $\Psi$ (Equation~\ref{eq:Psi}) is
governed by the following state space model:
\begin{align}
	\label{eq:extsys}
	\bmtx \dot{x}(t)\\\hline z(t)\\ e_{I}(t) \\ e_E(t)\emtx =
	\left[
	\begin{array}{c|c}
		\mathcal{A}(t) & \mathcal{B}(t)\\
		\hline
		\mathcal{C}_z(t) & \mathcal{D}_{z}(t)\\
		\mathcal{C}_{I}(t) & \mathcal{D}_{I}(t)\\
		\mathcal{C}_E(t) & 0 
	\end{array}
	\right]
	\bmtx x(t) \\ \hline \noalign{\vskip 0.05in}  \bmtx w(t)\\ d(t) \emtx\emtx
\end{align}
The extended state vector is $x:=\bsmtx x_N \\ x_\psi \esmtx \in \R^n$
where $n:=n_N+n_\psi$.  The state-space matrices are given by:
\begin{align*}
	&\mathcal{A} := \bmtx A_N & 0 \\  B_{\psi v} C_{v} & A_{\psi} \emtx,
	\mathcal{B} := \bmtx B_{w} & B_{d} \\ B_{\psi v} D_{vw} + B_{\psi w} & B_{\psi v} D_{vd}\emtx\\
	&\mathcal{C}_z := \bmtx D_{\psi v} C_{v} & C_{\psi} \emtx,
	\mathcal{C}_I := \bmtx C_{I} & 0 \emtx, \mathcal{C}_E := \bmtx C_E & 0 \emtx\\
	&\mathcal{D}_z := \bmtx D_{\psi v} D_{vw}  + D_{\psi w} &D_{\psi v} D_{vd}\emtx,
	\mathcal{D}_I = \bmtx D_{Iw} & D_{Id}\emtx
\end{align*}
The following differential linear matrix inequality (DLMI) is used to
compute an upper bound on the worst-case gain of
$\mathcal{F}_u(N,\Delta)$.
\begin{align}
	\label{eq:DLMI}
	\bmtx
	\dot{P}+\mathcal{A}^\top P+P\mathcal{A} & P\mathcal{B}\\
	\mathcal{B}^\top P & 0 
	\emtx + \bmtx
	Q & S\\
	S^\top & R
	\emtx
	+ \bmtx
	\mathcal{C}_z^\top \\ \mathcal{D}_z^\top
	\emtx M \bmtx
	\mathcal{C}_z & \mathcal{D}_z
	\emtx\leq -\epsilon I
\end{align}
This inequality depends on the IQC matrix $M$. It is compactly denoted as $DLMI_{Rob}(P,M,\gamma^2,t)\leq-\epsilon I$. This notation
emphasizes that the constraint is a DLMI in $(P,M,\gamma^2)$ for fixed
$N$, $\Psi$ and $(Q,S,R,F)$. The next theorem states a sufficient DLMI condition to bound the
generalized (robust) induced performance measure of
$\mathcal{F}_u(N,\Delta)$.  The proof is similar to Theorem~$6$ and $7$ of~\cite{seiler2019finite} and is given below for completeness. It uses IQCs~\cite{megretski1997system} and a standard dissipation argument~\cite{willems1972dissipative1,schaft1999,khalil2002nonlinear}.
\begin{theorem} 
	\label{thm:robperf}
	Consider an LTV system $N$ given by \eqref{eq:uLTV} and let
	$\Delta:\mathcal{L}^{n_v}_2[0,T]\rightarrow\mathcal{L}^{n_w}_2[0,T]$
	be an operator. Assume $\mathcal{F}_u(N,\Delta)$ is
	well-posed and $\Delta\in\mathcal{I}(\Psi,M)$. Let $Q : [0,T] \rightarrow \Sm^{n}$,
	$S : [0,T] \rightarrow \R^{n \times (n_w + n_d)}$,
	$R : [0,T] \rightarrow \Sm^{(n_w + n_d)}$, and
	$F \in \R^{n\times n}$ be defined as follows.
	\begin{align}
		\label{eq:QSRFRobPerf}
		\begin{split}
			Q  := \mathcal{C}_I^\top\mathcal{C}_I, \hspace{0.2in}  
			S := \mathcal{C}_I^\top\mathcal{D}_I, \hspace{0.2in} 
			R  := \mathcal{D}_I^\top\mathcal{D}_I - \gamma^2 \mbox{diag}\{0_{n_w}, I_{n_d}\}, \hspace{0.2in}  
			F  := \mathcal{C}_{E}(T)^\top\mathcal{C}_{E}(T)
		\end{split}
	\end{align}
	If there exists $\epsilon > 0$, $\gamma>0$ and a differentiable function $P:[0,T]\rightarrow\Sm^{n}$ such that $P(T)\geq F$ and,
	\begin{align}
		\label{eq:LMICond}
		DLMI_{Rob}(P,M,\gamma^2,t) \leq -\epsilon I \;\hspace{0.1in}\; \forall t\in[0,T]
	\end{align}
	then $\|\mathcal{F}_u(N,\Delta)\|_{[0,T]} < \gamma$.
\end{theorem}
\begin{proof}
	Let $d \in \Ltwo$ and $x_{N}(0)=0$ be given.  By
	well-posedness, $F_u(N,\Delta)$ has a unique solution $(x_N,v,w,e_I,e_E)$.
	Define $x := \bsmtx x_N \\ x_\psi \esmtx$.
	Then $(x,z,e_I,e_E)$ are a solution of the extended system
	\eqref{eq:extsys} with inputs $(w,d)$ and initial condition
	$x(0)= 0$.  Moreover, $z$ satisfies the the
	IQC defined by $(\Psi,M)$. Define a storage function by $V(x,t) := x^\top P(t) x$.  Left and right
	multiply the DLMI \eqref{eq:DLMI} by $[x^\top, w^\top, d^\top]$ and its
	transpose to show that $V$ satisfies the following dissipation
	inequality for all $t\in [0,T]$:
	\begin{align}
		\label{eq:DI}
		\dot{V} + 
		\bmtx x \\ \bsmtx w \\ d \esmtx \emtx^\top
		\bmtx Q & S \\ S^\top & R \emtx 
		\bmtx x \\ \bsmtx w \\ d \esmtx \emtx
		+ z^\top M z
		\le -\epsilon \, d^\top d.
	\end{align}
	Use the choices for $(Q,S,R)$ to rewrite the second term as $e_I^\top e_I - \gamma^2 d^\top d$.
	Integrate over $[0,T]$ to obtain:
	\begin{align*}
		x(T)^\top P(T)x(T) + \int_{0}^{T}z(t)^\top M(t)z(t) \, dt + \|e_I\|^2_{2,[0,T]} \leq (\gamma^2-\epsilon) \|d\|^2_{2,[0,T]}.
	\end{align*}
	Apply $P(T)\geq F = \mathcal{C}_{E}(T)^\top\mathcal{C}_{E}(T)$ and $\Delta \in \mathcal{I}(\Psi,M)$ to conclude:
	\begin{align}
		\begin{split}
			\label{eq:RobL2toL2_WithIC}
			\|e_E(T)\|_2^2 + \|e_I\|^2_{2,[0,T]} \leq (\gamma^2-\epsilon) \|d\|^2_{2,[0,T]}.
		\end{split}
	\end{align}
	This inequality implies $\|F_u(N,\Delta)\|_{[0,T]} < \gamma$.	
\end{proof}

\subsection{Computational Approach}
Numerical implementation using IQCs often involve a fixed choice of $\Psi$ and optimization subject to the convex constraints on $M$. Two examples are provided as follows.

\begin{ex}
	\label{ex:LTIuncParam}
	Consider an LTI uncertainty $\Delta \in \RH$ with
	$\|\Delta \|_\infty \le 1$. By Example~\ref{ex:LTIunc}, $\Delta$
	satisfies any IQC $(\Psi,M)$ with
	$\Psi :=\bsmtx \Psi_{11} & 0 \\ 0 & \Psi_{11} \esmtx$, 
	$M :=\bsmtx M_{11} & 0 \\ 0 & -M_{11} \esmtx$, and
	$M_{11} \succ 0$.  A typical choice for $\Psi_{11}$ is:
	\begin{align}
		\Psi_{11}:= \left[1, \frac{1}{(s+p)},\ldots \frac{1}{(s+p)^q}  \right]^\top
		\mbox{ with } p>0
	\end{align}  
	The analysis is performed by selecting $(p,q)$ to obtain (fixed)
	$\Psi$ and optimizing over the convex constraint $M_{11} \succ 0$. The results depend
	on the choice of $(p,q)$. Larger values of $q$ represent a richer
	class of IQCs and hence yield less conservative results but with
	increasing computational cost. Note that the IQC filter $\Psi$ is not square in general with $n_z = 2(q+1)$ outputs. 
\end{ex}
\begin{ex}
	Conic combinations of multiple IQCs can be incorporated in analysis.
	Let $(\Psi_i,M_i)$ with $i=1,2,\hdots,N$ define $N$ valid IQCs for
	$\Delta$. Hence $\int_0^T z_i^\top M_i z_i \, dt \ge 0$ where $z_i$ is
	the output $\Psi_i$ driven by $v$ and $w=\Delta(v)$. The multiple constraints can be multiplied by $\lambda_i\ge 0$ and combined to yield:
	\begin{align}
		\int_0^T \sum_{i=1}^{N}\lambda_i z_i^\top M_i z_i \, dt \ge 0
	\end{align}
	Thus a valid time-domain IQC for $\Delta$ is given by 
	\begin{align}
		\Psi:=\bmtx \Psi_1\\\vdots\\\Psi_N \emtx 
		\mbox{ and }
		M(\lambda):= \bmtx 
		\lambda_1 M_1 & & \\
		&\ddots& \\
		& & \lambda_N M_N
		\emtx
	\end{align}
	The analysis optimizes over $\lambda$ given selected $(\Psi_i,M_i)$.
\end{ex}
An iterative algorithm given in \cite{seiler2019finite} is used in this paper to compute the smallest upper bound on the worst-case gain. It combines the DLMI formulation in the Theorem~\ref{thm:robperf} with a related Riccati Differential Equation (RDE). The algorithm returns the upper bound $\bar\gamma_{wc}$ along with the decision variables $P$ and $M$.

\section{Robust Synthesis}
\label{sec:RobustSynthesis}

\subsection{Problem Formulation}
An uncertain feedback interconnection is shown in
\figref{fig:UncertainCL} where $G$ is an LTV system on $[0,T]$ and $\Delta$ is assumed to lie in some set $\mathcal{S}$ that is described by valid time domain IQCs. 
\begin{figure}[h]
	\centering
	\begin{tikzpicture}[thick,rounded corners = 0.5mm,scale=1.2]
		\begin{scope}[shift={(-0.2,0)}]
			\draw[densely dotted,fill=closedloop,fill opacity=1] (-0.7,-1.4) rectangle (2.3,1.4);
		\end{scope}
		\draw [fill=plant] (0,0) rectangle node{$G$}(1.2,1.2);
		\draw [->](0,1) -- (-0.5,1) -- (-0.5,2) -- (0.2,2);
		\draw [->](1,2) -- (1.7,2) -- (1.7,1) -- (1.2,1);
		\draw [->](0,0.8) -- (-1.2,0.8);
		\draw [->](0,0.4) -- (-1.2,0.4);
		\draw [->](2.4,0.6) -- (1.2,0.6);
		\node at (2.6,0.6) {$d$};
		\node at (1.9,1.7) {$w$};
		\node at (-0.7,1.7) {$v$};
		\node at (-1.5,0.8) {$e_I$};
		\node at (-1.5,0.4) {$e_E$};
		\draw [fill=controller] (0.2,-0.5) rectangle node{$K$}(1,-1.2);
		\draw [->](0,0.2) -- (-0.5,0.2) -- (-0.5,-0.8) -- (0.2,-0.8);
		\draw [->](1,-0.8) -- (1.7,-0.8) -- (1.7,0.2) -- (1.2,0.2);
		\node at (-0.7,-0.3) {$y$};
		\node at (1.9,-0.3) {$u$};
		\draw [fill=unc] (0.2,1.6) rectangle node{$\Delta$}(1,2.4);
		\node at (0.7,-1.7) {\footnotesize$N:=\mathcal{F}_l(G,K)$};
	\end{tikzpicture}	
	\caption{Uncertain Feedback Interconnection $\mathcal{F}_u(\mathcal{F}_l(G,K),\Delta)$}
	\label{fig:UncertainCL}
\end{figure}
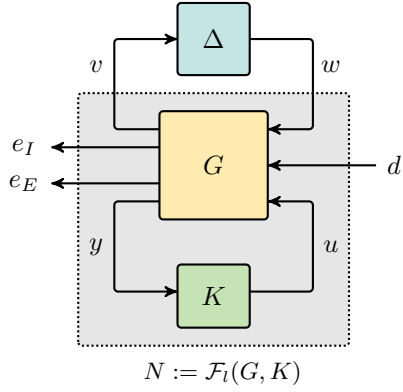\\
\noindent The finite horizon robust synthesis problem is to synthesize a controller which minimizes the impact of both worst-case disturbances and worst-case uncertainties, i.e.:
\begin{align}
	\label{eq:robsyn}
	\inf_{K} \sup_{\Delta \in \mathcal{S}} 
	\|\mathcal{F}_u(\mathcal{F}_l(G,K),\Delta)\|_{[0,T]}
\end{align}
Let the LTV system $G$ defined on $[0,T]$ be given as:
\begin{align}
	\label{eq:robLTVSyn}
	\bmtx \dot{x}_G(t)\\\hline v(t)\\ e_{I}(t) \\ e_E(t) \\y(t)\emtx =
	\left[
	\begin{array}{c|c c c}
		A_G(t) & B_w(t) & B_d(t) & B_u(t)\\
		\hline
		C_v(t) & D_{vw}(t) & D_{vd}(t) & D_{vu}(t)\\
		C_{I}(t) & D_{Iw}(t) & D_{Id}(t) & D_{Iu}(t)\\
		C_E(t) & 0 & 0 & 0\\
		C_y(t) & D_{yw}(t) & D_{yd}(t) & D_{yu}(t) 
	\end{array}
	\right]
	\bmtx x_G(t)\\\hline w(t)\\ d(t) \\ u(t)\emtx
\end{align}
where $x_G \in \R^{n_G}$ is the state. This plant structure has no feedthrough from $d$ to $e_E$ for
well-posedness. The synthesis problem \eqref{eq:robsyn}
involves the worst-case gain computed over the entire uncertainty set. As noted earlier, instead we focus on minimizing worst-case gain upper bounds. In other words,  we define IQCs $\mathcal{I}(\Psi,M)$ such that $\mathcal{S}\subseteq \mathcal{I}(\Psi,M)$ and maximize over  $\Delta \in \mathcal{I}(\Psi,M)$ in Equation~\eqref{eq:robsyn}. The goal is to design a linear time-varying controller $K:\mathcal{L}^{n_y}_2[0,T]\rightarrow\mathcal{L}^{n_u}_2[0,T]$ to minimize the worst-case gain upper bound on $\mathcal{F}_u(\mathcal{F}_l(G,K),\Delta)$. This leads to a non-convex synthesis problem and involves solving for the controller as well as IQC multipliers. 

The approach taken here is to decompose the synthesis into two subproblems. First, solve a nominal synthesis problem (on a specially constructed scaled plant) to obtain $K$. Second, solve an IQC analysis problem to compute the worst-case gain upper bound.  These subproblems can be solved iteratively, similar to coordinate descent, to get a reasonable sub-optimal solution. The proposed algorithm utilizes this approach to obtain a finite horizon sub-optimal controller. As with DK synthesis, there are no guarantees that the coordinate-wise iteration will lead to a local optima let alone a global optima. However, it is a useful heuristic that will enable the robust synthesis to extended naturally from LTI to finite horizon LTV systems. The following assumption is made for the structure of IQC matrix $M$ and filter $\Psi$.
\begin{asm}
	\label{asm:blkdiag}
	The IQC decision variables $M:[0,T]\rightarrow\Sm^{n_z}$ for a specified IQC filter $\Psi:\mathcal{L}_{2}^{(n_v+n_w)}[0,T] \rightarrow \mathcal{L}_{2}^{n_z}[0,T]$ are assumed to have the following block diagonal structure
	\begin{align*}
		M(t) := \bmtx M_{v}(t) & 0\\ 0 & -M_{w}(t)\emtx, \hspace{0.1in} \Psi := \bmtx \Psi_{v} & 0\\ 0 & \Psi_{w}\emtx
	\end{align*}
	with constraints $M_{v}(t)\succ 0$ and $M_{w}(t)\succ 0$, $\forall t\in[0,T]$. Moreover, $\Psi$ has a feedthrough matrix $D_\psi(t) := \bmtx D_{\psi v}(t) & D_{\psi w}(t)\emtx \in \R^{n_z \times (n_v+n_w)}$ with full column rank $\forall t\in[0,T]$.
\end{asm}	
This block diagonal assumption is made to simplify the notation. More general IQC multipliers are considered for (infinite horizon) synthesis in \cite{veenman2014iqc}. As discussed in Example~\ref{ex:LTIuncParam}, the IQC filter $\Psi$, is typically pre-specified by a collection of basis functions. In this case, the worst-case gain condition in Theorem~\ref{thm:robperf} is a differential LMI in the variables $M_v$, $M_w$, $P$, and $\gamma^2$.  The filter $\Psi$ is, in general, non-square with $n_z \neq n_v + n_w$. The proposed synthesis method requires a non-unique factorization such that resulting factor is invertible square system i.e. $n_z = n_v + n_w$. The finite horizon factorization (Lemma~\ref{lem:spectralfact} in Appendix~\ref{sec:spectralfact}) can be used to construct square invertible systems $U_v$ and $U_w$ such that,
\begin{align}
	\label{eq:PiSpectralFact}
	\begin{split}
		\Psi_v^\sim M_v \Psi_v = U_v^\sim U_v\\
		\Psi_w^\sim M_w \Psi_v = U_w^\sim U_w
	\end{split}
\end{align}
The assumption that feedthough matrix $D_\psi(t)$ has full column rank is required for the existence of such factorization. This factorization is used in the proposed synthesis algorithm below  to construct a scaled plant.

\subsection{Algorithm}
A high-level overview of the proposed iterative method is
given in Algorithm \ref{alg:fhrobsyn}. The uncertain finite horizon
system is $\mathcal{F}_u(\mathcal{F}_l(G,K),\Delta)$ with $G$ given by
Equation~\eqref{eq:robLTVSyn} and $\Delta$ specified by
uncertainty set $\mathcal{I}(\Psi,M)$. The robust synthesis algorithm is specified to run a given maximum number of iterations $N_{syn}$.  It is initialized with scalings
$U_{v}^{(0)}:=I_{n_v}$ and $U_{w}^{(0)}:=I_{n_w}$. There is also an
initial performance scaling set to $\gamma_a^{(0)}:=1$.
\begin{figure}[h]
	\centering
	\begin{tikzpicture}[thick,rounded corners = 0.5mm,scale=1.25]
		\begin{scope}[shift={(-0.2,0)}]
			\draw[densely dotted,fill=blue,fill opacity=0.1] (-1.3,-0.9) rectangle (3.8,1.4);
		\end{scope}
		\draw [fill=plant] (0,-0.7) rectangle node{$G$}(1.9,1.2);
		\draw [fill=dist] (2.5,-0.2) rectangle node{$\frac{1}{\gamma_a}$}(3.1,0.4);
		\draw [fill=unc] (2.4,0.6) rectangle node{$U_w^{-1}$}(3.3,1.2);
		\draw [fill=unc] (-1.3,0.6) rectangle node{$U_v$}(-0.41,1.2);
		\draw [->](0,0.4) -- (-1.9,0.4);
		\draw [->](0,0) -- (-1.9,0);
		\draw [->](0,0.9) -- (-0.42,0.9);
		\draw [->](2.5,0.1) -- (1.9,0.1);
		\draw [->](2.4,0.9) -- (1.9,0.9);
		\draw [->](3.8,0.9) -- (3.3,0.9);
		\draw [->](3.81,0.1) -- (3.1,0.1);
		\draw [->](-1.3,0.9) -- (-1.9,0.9);
		\node at (2.2,0.3) {$d$};
		\node at (4.01,0.1) {$\tilde{d}$};
		\node at (1,1.6) {$G_{scl}$};
		\node at (2.2,1.1) {$w$};
		\node at (4,0.9) {$\tilde{w}$};
		\node at (-0.2,1.1) {$v$};
		\node at (-2.2,0.9) {$\tilde{v}$};
		\node at (-2.2,0.4) {$e_I$};
		\node at (-2.2,0) {$e_E$};
		\draw [fill=controller,dashed] (0.6,-1.2) rectangle node{$K$}(1.4,-1.9);
		\draw [->](0,-0.5) -- (-0.4,-0.5) -- (-0.4,-1.5) -- (0.6,-1.5);
		\draw [->](1.4,-1.5) -- (2.3,-1.5) -- (2.3,-0.5) -- (1.9,-0.5);
		\node at (2.5,-1.2) {$u$};
		\node at (-0.6,-1.2) {$y$};
	\end{tikzpicture}	
	\caption{LTV Synthesis on Scaled Plant $G_{scl}$}
	\label{fig:NomLTVSynStep}
\end{figure}
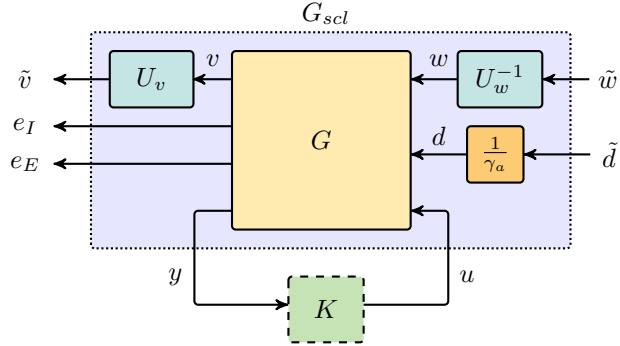

The beginning of each iteration involves the construction of a scaled plant
$G_{scl}$ as shown in \figref{fig:NomLTVSynStep}. This step is described further in the next subsection. For now, it is sufficient to note that $G_{scl}=G$ on the first iteration due to the
initialization choices.
\begin{algorithm}[t]	
	\linespread{1}\selectfont
	\caption{Finite Horizon Robust Synthesis} \label{alg:fhrobsyn}  
	\begin{algorithmic}[1]
		\State \textbf{Given:} $G$ 
		\State \textbf{Initialize:} $N_{syn}$, $U_{v}^{(0)}:=I_{n_v}$,
		$U_{w}^{(0)} := I_{n_w}$, $\gamma_a^{(0)} := 1$
		
		\For{$i=1:N_{syn}$} 
		
		\State \parbox[t]{0.9\linewidth}{\textbf{Scaled Plant
				Construction (Section~\ref{sec:scaledplantconstruction}):} Construct a scaled plant $G_{scl}^{(i)}$ using $G$, $U_{v}^{(i-1)}$,
			$U_{w}^{(i-1)}$, $\gamma_a^{(i-1)}$.\\\textbf{Output:} $G_{scl}^{(i)}$}	
		
		\vspace{0.1in}
		\Statex
		
		\State \parbox[t]{0.9\linewidth}{ \textbf{Nominal LTV Synthesis
				(Section~\ref{sec:nominalsyn}):} Perform nominal controller synthesis
			on the scaled plant $G_{scl}^{(i)}$.\\\textbf{Output:} $K^{(i)}$, $\gamma_s^{(i)}$}
		
		\vspace{0.1in}	
		\Statex 
		
		\State \parbox[t]{0.9\linewidth}{\textbf{IQC Analysis
				(Section~\ref{sec:robperf}):} Choose the basis functions for $\Psi$ and perform worst-case gain iterations on $\mathcal{F}_u(N^{(i)},\Delta)$ using iterative algorithm	presented in \cite{seiler2019finite} where $N^{(i)}:=\mathcal{F}_l(G,K^{(i)})$ denotes the closed loop LTV system. Perform finite horizon factorization using the same $\Psi$ and computed decision variables $M^{(i)}$ to compute the uncertainty channel scalings $U_v^{(i)}$ and $U_w^{(i)}$.\\	\textbf{Output:} $P^{(i)}$, $M^{(i)}$, $\gamma_a^{(i)},U_v^{(i)},U_w^{(i)}$}
		
		\Statex		
		\EndFor
	\end{algorithmic}
\end{algorithm}
The next step is to perform finite horizon
nominal synthesis on the scaled plant.  This step is performed using
the synthesis results described previously in
Section~\ref{sec:nominalsyn}. This yields a controller $K^{(i)}$ and
the achievable closed-loop performance $\gamma_s^{(i)}$. Each
iteration concludes with an IQC analysis on the uncertain closed-loop
of $N:=\mathcal{F}_l(G,K^{(i)})$ and $\Delta$ as shown in
\figref{fig:UncertainCL}. This closed-loop uses the original
(unscaled) plant $G$ and the controller $K^{(i)}$ obtained from the
nominal synthesis step. The worst-case gain upper bound
$\gamma_a^{(i)}$ is computed using the algorithm in
\cite{seiler2019finite} as summarized in
Section~\ref{sec:robperf}. This iterative algorithm requires additional initialization including number of analysis iterations $N_{iter}$, stopping tolerance $tol$, DLMI time grid $t_{DLMI}$ and spline basis function time grid $t_{sp}$, which are not included in Algorithm \ref{alg:fhrobsyn}. All subsequent iterations require the construction of a scaled plant using the IQC results. The construction of this scaled plant links together the nominal synthesis and IQC analysis steps.  It is described further in Section~\ref{sec:scaledplantconstruction}. Algorithm \ref{alg:fhrobsyn} terminates after $N_{syn}$
iterations. More sophisticated stopping conditions can be employed.  For example, the iterations could be terminated if no significant
improvement in worst-case gain is achieved. The algorithm returns the
controller of order $n_K$ that achieves the best (smallest) bound on the worst-case gain, where $n_K = n_G + n_\psi$.

\subsection{Construction of a Scaled Plant}
\label{sec:scaledplantconstruction}

The scaled open loop plant $G_{scl}^{(i)}$ is constructed as shown in
\figref{fig:NomLTVSynStep} by scaling the performance channels and
uncertainty channels of original open loop plant $G$ using $U_{v}^{(i-1)}$,
$U_{w}^{(i-1)}$ and $\gamma_a^{(i-1)}$ obtained from the previous
iteration. This scaling ensures appropriate normalization of the
performance and uncertainty channels. This is a key step which
integrates the nominal synthesis and worst-case gain problem. To simplify the notation, the superscripts $(i-1)$ will be 
dropped in the remainder of this subsection. 
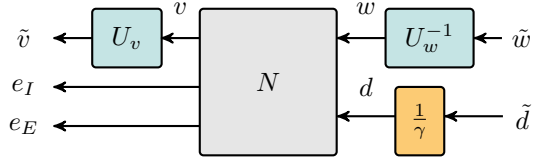
\begin{figure}[h]
	\centering
	\begin{tikzpicture}[thick,rounded corners = 0.5mm,scale=1.3]
		\draw [fill=closedloop] (0,-0.3) rectangle node{$N$}(1.4,1.2);
		\draw [fill=dist] (2,-0.3) rectangle node{$\frac{1}{\gamma}$}(2.5,0.4);
		\draw [fill=unc] (1.9,0.6) rectangle node{$U_{w}^{-1}$}(2.8,1.2);
		\draw [fill=unc] (-1.1,0.6) rectangle node{$U_{v}$}(-0.39,1.2);
		\draw [->](0,0.4) -- (-1.5,0.4);
		\draw [->](0,0) -- (-1.5,0);
		\draw [->](0,0.9) -- (-0.4,0.9);
		\draw [->](2,0.1) -- (1.4,0.1);
		\draw [->](1.9,0.9) -- (1.4,0.9);
		\draw [->](3.09,0.9) -- (2.8,0.9);
		\draw [->](3.1,0.1) -- (2.5,0.1);
		\draw [->](-1.1,0.9) -- (-1.5,0.9);
		\node at (1.7,0.4) {$d$};
		\node at (3.3,0.1) {$\tilde{d}$};
		\node at (1.7,1.2) {$w$};
		\node at (3.29,0.9) {$\tilde{w}$};
		\node at (-0.2,1.2) {$v$};
		\node at (-1.8,0.9) {$\tilde{v}$};
		\node at (-1.8,0.4) {$e_I$};
		\node at (-1.8,0) {$e_E$};
	\end{tikzpicture}
	\caption{Scaled Plant $N_{scl}$}
	\label{fig:ScaledPlant}
\end{figure}

Let $N:=\mathcal{F}_l(G,K)$ be the closed-loop (without uncertainty). For a given IQC filter $\Psi$ an extended system $N_{ext}$ similar to \figref{fig:IQCAugment} can be constructed. The next lemma gives a formal statement connecting robust performance of the extended system $N_{ext}$ to nominal performance of scaled system $N_{scl}$ as shown in \figref{fig:ScaledPlant}.

\begin{lem}
	\label{lm:scaledplant}
	Let $\epsilon > 0$, $\gamma>0$, $M_v(t)\succ 0$, $M_w(t)\succ 0$ and a differentiable function $P:[0,T]\rightarrow\Sm^{n}$ such that $P(T)\geq F$ be given with the choice of $(Q,S,R,F)$ as in Equation~\eqref{eq:QSRFRobPerf}. The following statements are equivalent:
	\begin{enumerate}
		\item $DLMI_{Rob}(P,M,\gamma^2,t) \leq -\epsilon I, \;\hspace{0.1in}\; \forall t\in[0,T].$
		\item $\|N_{scl}\|_{[0,T]} \leq 1 - \hat{\epsilon}$,\hspace{0.1in} for some $\hat{\epsilon}$.
	\end{enumerate}
\end{lem}
\begin{proof}
	A proof of this lemma is given in Appendix~\ref{sec:proof1}. It uses a time-varying factorization of $(\Psi,M)$ to construct $U_v$ and $U_w$.
\end{proof}
\noindent The above lemma states that extended system given by Equation~\eqref{eq:extsys} satisfies the robust performance condition \eqref{eq:LMICond} if and only if the scaled system has nominal performance less than 1. 

\subsection{Main Theorem}
\label{sec:mainth}

The plant $G_{scl}^{(1)} = G$ for the robust synthesis may include the	uncertainty and performance channel design weights as in standard robust control workflow \cite{zhou1996robust,dullerud2013course}. These weights can be static, dynamic and/or time-varying depending on the requirement. Typically, multiple design iterations are performed to tune these weights and yield an acceptable trade-off between robustness and performance. Note that the first nominal LTV synthesis step in Algorithm \ref{alg:fhrobsyn} may not yield a finite performance $\gamma_s^{(1)}$. For example, if the uncertainty level is too high then the RDEs for nominal synthesis of $G_{scl}^{(1)} = G$ may not have a solution on $[0,T]$ for any finite $\gamma_s^{(1)}$. However, in this case, finite performance can be achieved by reducing the uncertainty level and restarting the iteration. The main theorem is presented next with a technical assumption that the first nominal synthesis step yields a finite performance.

\begin{theorem}
	If the first nominal synthesis step yields a finite performance $\gamma_s^{(1)}$ then all the subsequent iterations are well-posed at each step and worst-case gain is non-increasing, i.e.
	\begin{align*}
		\gamma_a^{(i+1)}\leq\gamma_a^{(i)} \hspace{0.2in} \forall{i} \geq 1
	\end{align*} 
\end{theorem}
\begin{proof}
	The first iteration $(i=1)$ is different from the subsequent one. Due to initialization choices $G_{scl}^{(1)}=G$. The synthesis step is performed with no modifications and yields a controller $K^{(1)}$ that guarantees the closed loop performance of $\gamma_s^{(1)}$. By assumption, we have $\gamma_s^{(1)}<\infty$. The IQC analysis step performed on the closed loop $N^{(1)} := \mathcal{F}_l(G,K^{(1)})$ uncertain plant then achieves a finite horizon worst-case gain upper bound of $\gamma_a^{(1)}<\infty$. Thus, the first iteration is well-posed.
	\vspace{0.1in}
	
	All subsequent iterations $(i>1)$ begin with the iteration count update in the for loop. The IQC analysis step from previous iteration shows that there exists ($P^{(i-1)}$, $M^{(i-1)}$, $\gamma_a^{(i-1)}$) for a chosen $\Psi$ that satisfies DLMI \eqref{eq:DLMI}. This implies that the finite horizon factorization exists and multipliers $U_v^{(i-1)}$ and $U_w^{(i-1)}$ can be obtained using Lemma~\ref{lem:spectralfact} in Appendix~\ref{sec:spectralfact}. Using these multipliers and worst-case gain $\gamma_a^{(i-1)}$, scaled plant similar to \figref{fig:ScaledPlant} can be constructed. By Lemma~\ref{lm:scaledplant}, this scaled plant satisfies nominal performance $< 1$. Removing the controller yields the scaled open-loop plant $G_{scl}^{(i)}$. Thus, the construction of a scaled open-loop plant as shown in \figref{fig:NomLTVSynStep} is well-defined. The synthesis step performed on $G_{scl}^{(i)}$ optimizes over all time-varying finite horizon controllers to yield a new controller $K^{(i)}$ that guarantees performance $\gamma_s^{(i)} < 1$. This new controller $K^{(i)}$ yields better nominal performance than the previous controller $K^{(i-1)}$ when used with the unscaled plant $G$. Thus, the closed loop $N^{(i)} :=\mathcal{F}_u(G,K^{(i)})$ must satisfy the nominal performance $<1$ when using $\gamma_a^{(i-1)}$. Lemma~\ref{lm:scaledplant} can be used backwards in the next analysis step of $N^{(i)}$. Specifically, the closed loop with unscaled plant $G$ and $K^{(i)}$ satisfies the DLMI analysis condition with ($P^{(i-1)}$, $M^{(i-1)}$, $\gamma_a^{(i-1)}$). Further, analysis step on $N^{(i)} := \mathcal{F}_l(G,K^{(i)})$ optimizes over all feasible $P$ and $M$. This yields a worst-case gain $\gamma_a^{(i)}$ no greater than the previous step $\gamma_a^{(i-1)}$. Thus $\forall{i}\geq 1$, we have $\gamma_a^{(i+1)}\leq\gamma_a^{(i)}$.
\end{proof}

\section{Numerical Examples}
\label{sec:examples}

\subsection{LTI Example}
\label{sec:LTIEx}

Consider a first order linear time-invariant (LTI) system $G$ with the following dynamics:
\begin{align}
	\dot{x}(t) &= 0.5\, x(t) +  u(t) + w(t) + d_{in}(t)\\
	v(t) &= u(t) + d_{in}(t)\\
	e_I(t) &= \bmtx x(t)\\ 0.2 \, u(t) \emtx\\
	y(t) &= x(t) + 0.01\, n(t)
\end{align}
where the performance output is $e_I(t) \in \R^2$ and measurement output is $y(t) \in \R$. The generalized disturbance input is $d(t):=\ \bsmtx d_{in}(t)\\ n(t)\esmtx$ where $d_{in}(t)\in\R$ is an external disturbance input and $n(t)\in\R$ is a measurement noise. $\Delta$ is a norm bounded, time-varying, nonlinear uncertainty with norm bound $\beta = 0.6$. The goal is to design a measurement feedback controller $K$ that minimizes the worst-case induced $\mathcal{L}_2$-gain from disturbance $d$ to output $e_I$.

First, an infinite horizon robust synthesis is performed to minimize the worst-case gain of the closed loop system. This is achieved by performing bisection on the gain $\gamma$ for a scaled plant until the robust performance is equal to $1$. At each bisection step, MATLAB's \texttt{musyn} function is called to optimize the structured singular value upper bound $\bar{\mu}$. This function uses an iterative control design process (DK-iteration) to optimizes the $\bar{\mu}$ of the closed loop system. The infinite horizon worst-case induced $\mathcal{L}_2$-gain for the designed robust controller is $0.0563$. Next, finite horizon robust synthesis is performed on a relatively long horizon (i.e. $T = 300$ seconds) using the method proposed in this paper. The closed-loop worst-case gain achieved by a finite horizon time-varying controller is computed as $0.0560$. This simple comparison results show that on a relatively long horizon the worst-case gain achieved using the finite horizon controller approaches to that of the value achieved by an infinite horizon controller. Note that the proposed method uses purely time-domain approach whereas the $\mu$-synthesis method uses the frequency gridding approach to approximate $\bar{\mu}$ and the associated $D$-scales. Thus, the close agreement between the two worst-case gains on this example may not hold in general.

\subsection{Nonlinear Example}
\label{sec:NLEx}

Consider an example of a two-link robot arm as shown in the \figref{fig:twoLinkRobot}. The mass and moment of inertia of the $i^{th}$ link are denoted by $m_i$ and $I_i$. The robot properties are $m_1=3kg$,
$m_2 = 2kg$, $l_1 = l_2= 0.3m$, $r_1=r_2 = 0.15m$,
$I_1= 0.09 kg\cdot m^2$, and $I_2= 0.06 kg\cdot m^2$. The nonlinear equations of motion \cite{murray1994mathematical} for the robot are given by:
\begin{eqnarray}
	\nonumber
	&&  \bmat{\alpha+ 2\beta\cos(\theta_2) & \delta + \beta \cos(\theta_2) \\
		\delta +  \beta \cos(\theta_2) & \delta}
	\bmat{\ddot{\theta}_1 \\ \ddot{\theta}_2} +  \bmat{-\beta \sin(\theta_2) \dot{\theta}_2 &
		-\beta \sin(\theta_2) (\dot{\theta}_1 + \dot{\theta}_2) \\
		\beta \sin(\theta_2) \dot{\theta}_1 & 0} \bmat{\dot{\theta}_1 \\
		\dot{\theta}_2} =\bmat{\tau_1 \\ \tau_2}  \label{eq:linkDyns}
	\\
	\nonumber
	&& \mbox{with}
	\\
	\nonumber
	&& \alpha := I_{1} + I_{2} + m_1r_1^2 + m_2(l_1^2 + r_2^2) = 0.4425    \, kg \cdot m^2\\
	\nonumber
	&& \beta := m_2 l_1 r_2 = 0.09 \, kg \cdot m^2 \\
	\nonumber
	&& \delta := I_{2} + m_2 r_2^2 = 0.105 \, kg \cdot m^2.
\end{eqnarray}
The state and input are
$\eta :=[\theta_1  \ \theta_2  \ \dot{\theta}_1 \
\dot{\theta}_2]^\top$ and $\tau :=[\tau_1 \ \tau_2]^\top$,
where $\tau_i$ is the torque applied to the base of link $i$. 
\begin{figure}[h] 
	\centering
	\includegraphics[scale=0.35]{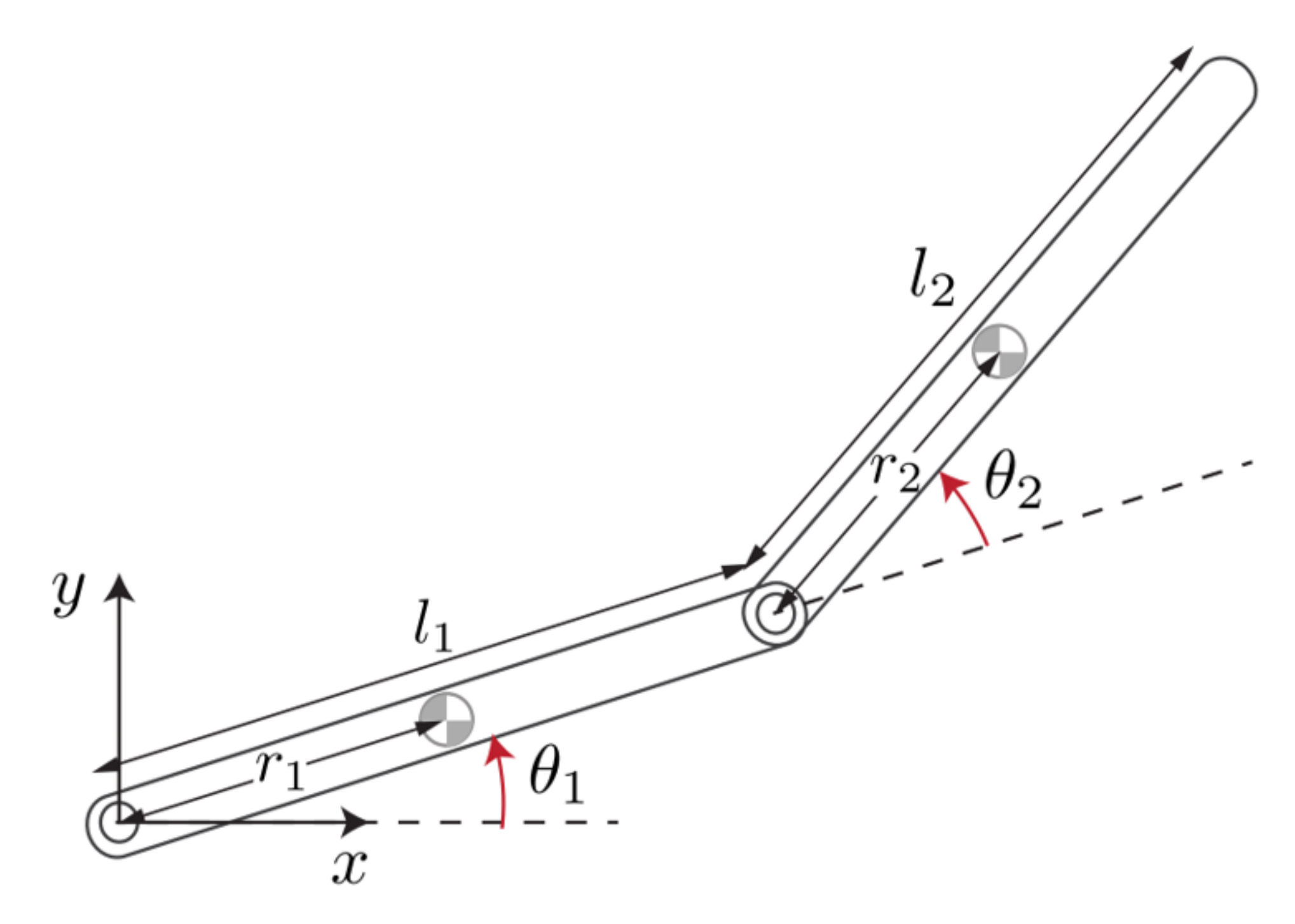}
	\caption{Two-link Planar Robot \cite{murray1994mathematical}.}
	\label{fig:twoLinkRobot}
\end{figure}

A trajectory $\bar{\eta}$ of duration $5$ second was selected for the tip of the arm to follow. This trajectory is shown as a solid black line in \figref{fig:NominalTrajectory1}. An equivalent trajectory in polar coordinates is also shown in \figref{fig:NominalTrajectory2}. The equilibrium input torque $\bar{\tau}$ can be computed using inverse kinematics.
\begin{figure}[h]
	\centering
	\begin{minipage}{.5\textwidth}
		\centering
		\includegraphics[width=\linewidth]{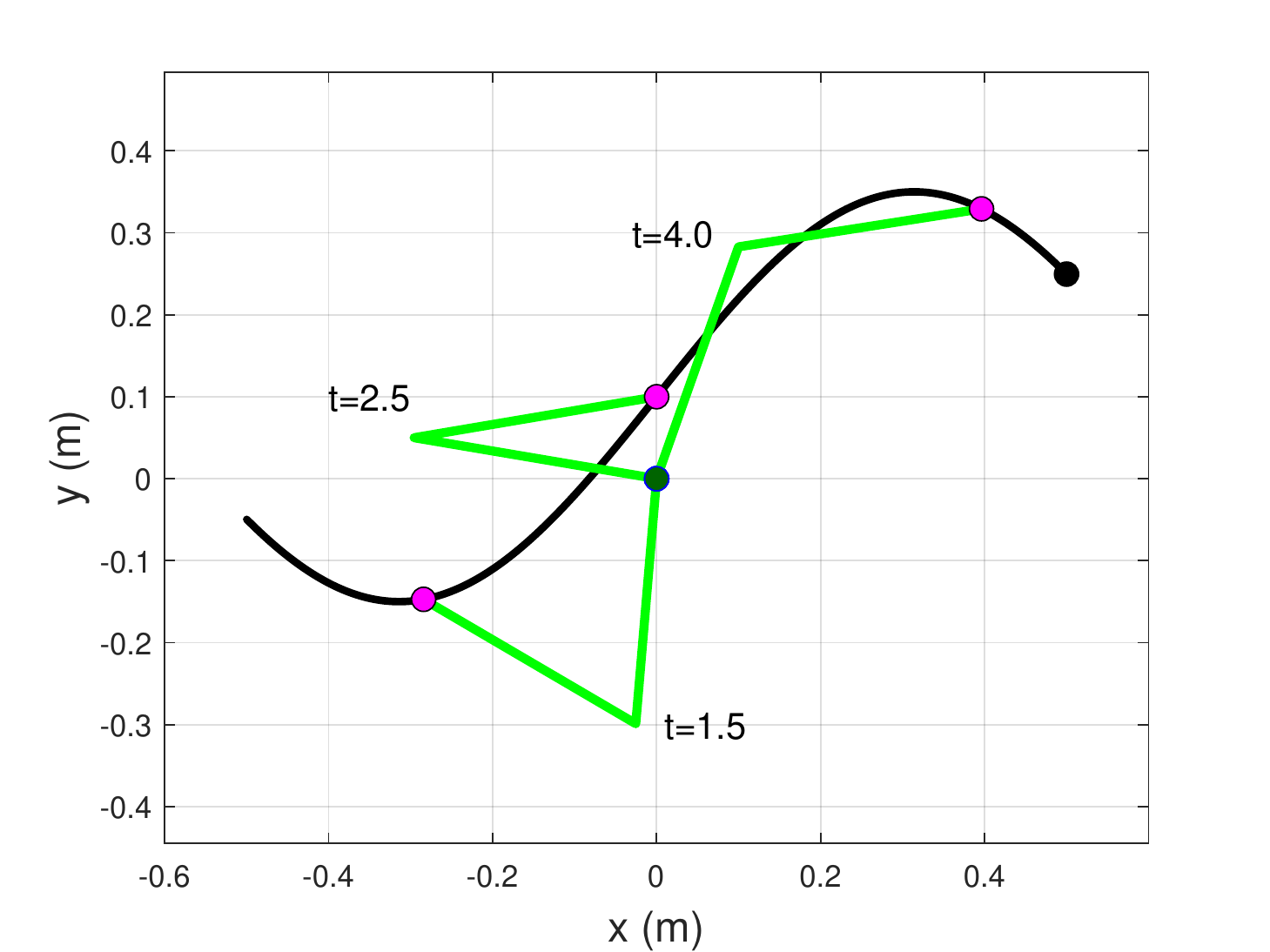}
		\caption{Snapshot Positions in Cartesian Coordinates}
		\label{fig:NominalTrajectory1}
	\end{minipage}%
	\begin{minipage}{.5\textwidth}
		\centering
		\includegraphics[width=\linewidth]{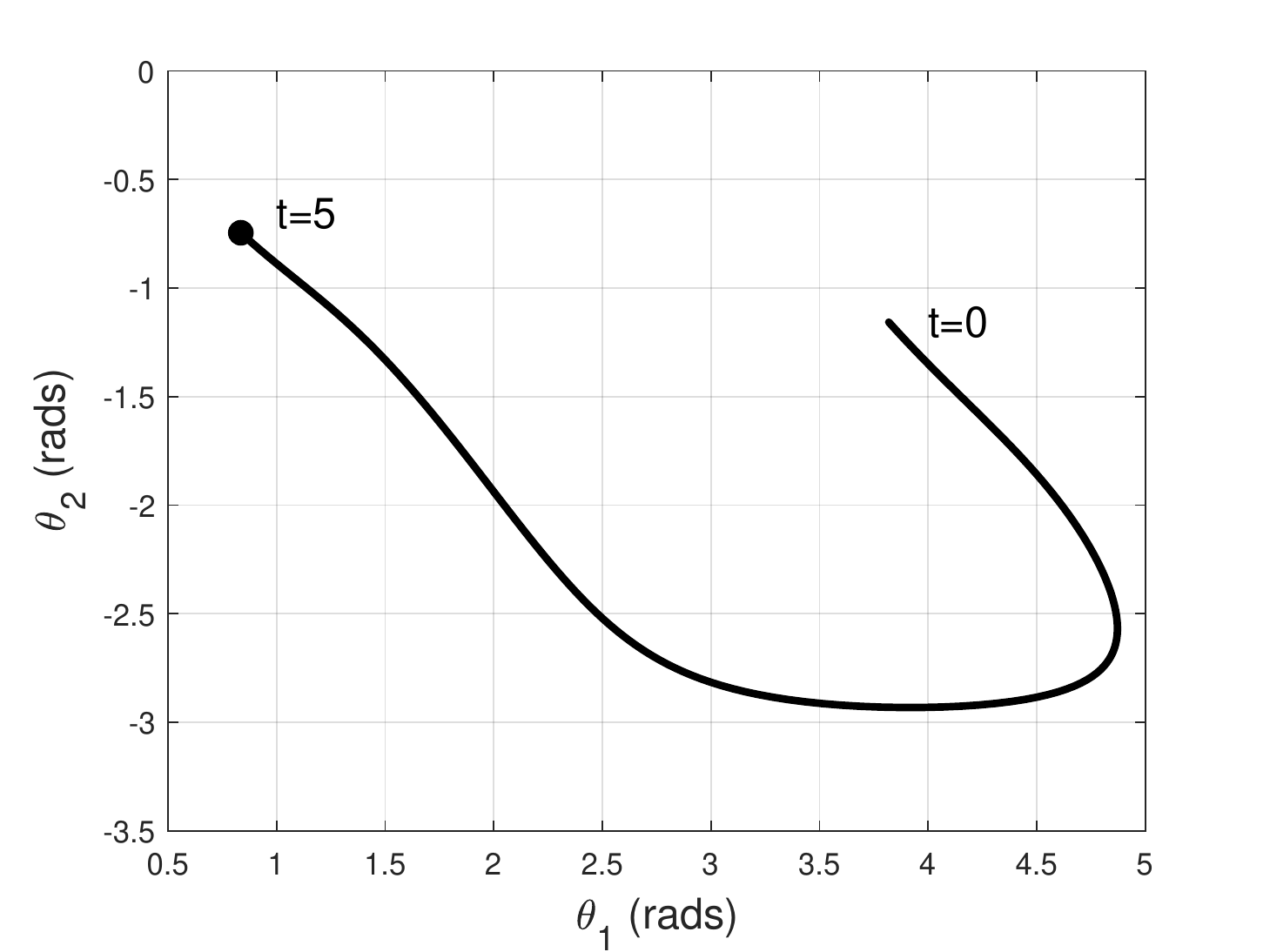}
		\caption{Nominal Trajectory in Polar Coordinates}
		\label{fig:NominalTrajectory2}
	\end{minipage}
\end{figure}
The robot should track this trajectory in the presence of small torque
disturbances $d_{in}$. The input torque vector is
$\tau = \bar{\tau} + u + d_{in}$ where $u$ is an additional control torque to reject the disturbances. The nonlinear dynamics
\eqref{eq:linkDyns} are linearized around the trajectory
$(\bar{\eta}, \bar{\tau})$ to obtain an LTV system $H$:
\begin{align}
	\dot{x}(t) = A(t) \, x(t) + B(t) \left( u(t) + d_{in}(t) \right)
\end{align}
where $x(t):=\eta(t)-\bar{\eta}(t)$ is the deviation from the
equilibrium trajectory. An uncertain output feedback weighted interconnection of $H$ is shown in the \figref{fig:SignalWeightedDesign}. Let $\delta\theta := \bsmtx \delta\theta_1 \\ \delta\theta_2 \esmtx$ represent first-order perturbations in angular positions, which is the output of interest $e_E$. The measurement is also $\delta\theta$ but corrupted by noise $n=\bsmtx n_1 \\ n_2 \esmtx$ and is fed to the controller as $y = \delta\theta + n$. The controller generates a commanded
torque $u = \bsmtx u_1 \\ u_2 \esmtx$ is corrupted by input disturbance $d_{in} = \bsmtx d_1\\ d_2\esmtx$. The second control channel gets further corrupted by LTI input uncertainty $\Delta$.
\begin{figure}[t]
	\centering
	\begin{tikzpicture}[thick,scale=1.3,rounded corners = 0.5mm,every node/.style={scale=1}]
		\draw [fill=plant] (-2.28,0) rectangle node{$H$}(-1.3,0.9);
		\draw [->](-1.3,0.7) -- (-0.5,0.7);
		\draw [->](-4.2,0.7) -- (-2.3,0.7);
		\draw [-](-4.5,0.2) -- (-4.7045,0.2);
		\draw [-](-4.4,-1.1) -- (-4.8,-1.1);
		\draw [->](-5.62,0.7) -- (-4.61,0.7);
		\draw [->](-4.3,0.2) -- (-3.1415,0.2);
		\draw [->](-2.7415,0.2) -- (-2.3,0.2);
		\draw [->](-5.6,0.2) -- (-5.1045,0.2);
		\draw [->](-6.92,0.7) -- (-6.32,0.7);
		\draw [->](-6.92,0.2) -- (-6.3,0.2);
		\draw [->](-1.3,0.2) -- (-0.9,0.2) -- (-0.9,-1.1);
		\draw [->](-0.2,-1.3) -- (-0.7,-1.3);
		\draw [->](0.9,-1.3) -- (0.5,-1.3);
		\draw [->](-4.0515,0.2) -- (-4.0415,-0.4) -- (-3.7315,-0.4);
		\draw [->](-1.1,-1.3) -- (-1.7,-1.3);
		\draw [->](0.2,0.7) -- (0.8,0.7);
		\draw [->](-5,-1.1) -- (-5.6,-1.1);
		\draw [->](-4.9,-1.6) -- (-5.6,-1.6);
		\draw [->](-6.3,-1.1) -- (-6.9,-1.1);
		\draw [->](-6.3,-1.6) -- (-6.9,-1.6);
		\draw [->](-3.1415,-0.4) -- (-2.9415,-0.4) -- (-2.9415,0);
		\node at (0.5,1) {$\tilde{e}_E$};
		\node at (-0.9,1) {$e_E$};
		\node at (-1.4,-1.6) {$y$};
		\node at (-5.3,-0.9) {$u_1$};
		\node at (-5.3,-1.4) {$u_2$};
		\node at (-6.7,-0.8732) {$\tilde{u}_1$};
		\node at (-6.7,-1.4) {$\tilde{u}_2$};
		\node at (0.8,-1) {$\tilde{n}$};
		\node at (-4.2,-0.4) {$v$};
		\node at (-2.7415,-0.4) {$w$};
		\node at (-6.7,0.9) {$\tilde{d}_1$};
		\node at (-6.7,0.4) {$\tilde{d}_2$};
		\node at (-5.4,0.9) {$d_1$};
		\node at (-5.4,0.4) {$d_2$};
		\node at (-0.46,-1) {$n$};
		\node at (-4.68,0.5) {$+$};
		\node at (-1.2,-1) {$+$};
		\node at (-2.7415,-0.1) {$+$};
		\node at (-5.2045,0) {$+$};
		\draw [fill=controller,dashed] (-2.6,-0.9) rectangle node{$K$}(-1.7,-1.8);
		\draw [fill=unc] (-3.7415,-0.1) rectangle node{$\Delta$}(-3.1415,-0.7);
		\draw [fill=dist] (-6.32,0.9) rectangle node{$W_{d}$}(-5.62,0);
		\draw [fill=dist] (-0.5,1) rectangle node{$W_E$}(0.2,0.4);
		\draw [fill=dist] (-6.3,-0.8732) rectangle node{$W_u$}(-5.6,-1.8);
		\draw [fill=dist] (-0.2,-1) rectangle node{$W_n$}(0.5,-1.6);
		\draw [->](-2.6,-1.1) -- (-4.4,-1.1) -- (-4.4,0.5);
		\draw [->](-2.6,-1.6) -- (-4.9,-1.6) -- (-4.9045,0);
		\draw  (-4.4,0.7) ellipse (0.2 and 0.2);
		\draw  (-4.9,0.2) ellipse (0.2 and 0.2);
		\draw  (-0.9,-1.3) ellipse (0.2 and 0.2);
		\draw  (-2.9415,0.2) ellipse (0.2 and 0.2);
		\draw [fill=black] (-4.05,0.2) ellipse (0.05 and 0.05);
		\draw [fill=black] (-4.4,-1.1) ellipse (0.05 and 0.05);
		\draw [fill=black] (-4.9,-1.6) ellipse (0.05 and 0.05);
		\draw (-4.3,0.2) arc (0:-180:0.1);
		\draw (-5,-1.1) arc (-180:0:0.1);
	\end{tikzpicture}
	\caption{Uncertain Output Feedback Weighted Interconnection}
	\label{fig:SignalWeightedDesign}
\end{figure}
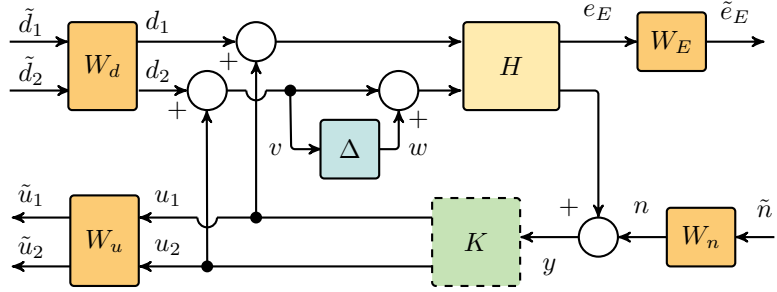
The plant input perturbation $\Delta$ is
a SISO LTI system with $\|\Delta\|_\infty \le \beta$ where uncertainty level $\beta = 0.8$. This corresponds to the uncertainty set as discussed in
Example~\ref{ex:LTIunc}. The synthesis objective is to minimize the closed-loop, worst-case
gain from the generalized disturbance
$\tilde{d}:=\bsmtx \tilde{d}_{in} \\ \tilde{n}\esmtx = \bsmtx \tilde{d}_1 & \tilde{d}_2 & \tilde{n}_1 & \tilde{n}_2 \esmtx^\top$ to the generalized error $\tilde{e}:=\bsmtx \tilde{u} \\ \tilde{e}_E \esmtx$.  The weighted control
effort $\tilde{u}$ is penalized in an $\mathcal{L}_2[0,T]$ sense while
$\tilde{e}_E$ is penalized with a terminal Euclidean norm at
$T=5$ second. Let $I_2 := \bsmtx 1 & 0\\ 0 & 1\esmtx$. The following (constant) design weights are chosen for the performance channels:
\begin{align*}
	W_d = 0.1 \, I_2, \,\,\,\,\, W_n = 0.01 \, I_2, \,\,\,\,\, W_u = 0.5 \, I_2, \,\,\,\,\, W_E = I_2
\end{align*}
The design weight associated with the uncertainty channels are not considered in this example, however, in general, the weights $W_v$ and $W_w$ can also be used for the respective uncertainty channels. As noted earlier, these design weights can be dynamic and/or time-varying. Let $\tilde{G}$ denote this weighted design interconnection for robust synthesis. It can be expressed in state-space form as in Equation~\eqref{eq:robLTVSyn}. Algorithm~\ref{alg:fhrobsyn} is run with $N_{syn} = 7$ iterations. No significant improvement is obtained after $7^{th}$ iteration. The IQC analysis step is performed based on the approach in \cite{seiler2019finite} and using parameterization similar to Example~\ref{ex:LTIuncParam} with $p = 10$, $q = 1$, $tol = 5 \times 10^{-3}$, $N_{iter} = 10$, $t_{DLMI}$ as $20$ and $\tau_{sp}$ as $10$ evenly spaced grid points on the horizon $[0,5]$ seconds. 

Let $K_{0.8}$ denote the controller obtained at the end of the robust
synthesis algorithm. This controller achieves the closed-loop
worst-case performance of $\gamma_{0.8}=0.126$. It took $11.7$ hours to complete the $7$ iterations on a standard desktop computer with $3$ GHz Core i7 processor. 
In addition, a nominal synthesis with $\Delta=0$ was performed using the
approach in Section \ref{sec:nominalsyn}. This controller, denoted
as $K_0$, achieves a closed-loop nominal performance of
$\gamma_0 = 0.089$. It took $49.8$ seconds to perform this nominal synthesis. The corresponding uncertain closed-loops with the nominal
and robust controllers are denoted by
$\tilde{T}_0 := \mathcal{F}_u(\mathcal{F}_l(\tilde{G},K_0),\Delta)$
and
$\tilde{T}_{0.8} :=
\mathcal{F}_u(\mathcal{F}_l(\tilde{G},K_{0.8}),\Delta)$. \figref{fig:CompareBothControllers} shows the worst-case performance versus
the uncertainty level $\beta$ for the uncertain closed-loops with
these two controllers. The curve for $\tilde{T}_0$ (blue circles) has
$\gamma=0.089$ at $\beta=0$ as reported above. The curve for $\tilde{T}_{0.8}$
(red squares) has $\gamma=0.126$ at $\beta=0.8$ as also reported
above. This figure reveals the typical trade-off between performance
and robustness. The nominal controller $K_0$ achieves better nominal
performance ($\beta=0$) than $K_{0.8}$. However, $K_{0.8}$ is more
robust to higher levels of uncertainties.

Note that each data point in~\figref{fig:CompareBothControllers} represents a worst-case gain induced from the generalized disturbance input $\tilde{d}$ to the generalized error output $\tilde{e}$. Both $\tilde{d}$ and $\tilde{e}$ have two components which can further be analyzed using induced gain for individual input-output pairs. Table~\ref{table:1} shows an analysis with no uncertainty (level $\beta=0$) performed for the closed-loops with nominal control design $\tilde{T}_0$ (blue) and robust control design $\tilde{T}_{0.8}$ (red). It is evident that the induced gain from noise $\tilde{n}$ to control effort $\tilde{u}$ dominates the overall performance for both the interconnections. Moreover, the induced $\mathcal{L}_2$-gain from $\tilde{n}$ to $\tilde{u}$ for $\tilde{T}_{0.8}$ is $0.124$ which is approximately $41\%$ higher than the corresponding value for $\tilde{T}_{0}$ ($=0.088$). Likewise, the induced $\mathcal{L}_2$-to-Euclidean gain from disturbance $\tilde{d}_{in}$ to $\tilde{e}_E(T)$ is approximately $16.6\%$ higher for $\tilde{T}_{0.8}$ ($=0.077$) as compared to $\tilde{T}_{0}$ ($=0.066$). The combined effect of the disturbance and noise is responsible for performance degradation of the robust controller at $\beta = 0$. Table~\ref{table:2} shows the worst-case gain upper bounds for robust analysis performed at an uncertainty level $\beta = 0.8$ for both interconnections. Note that the closed-loop with robust controller has the same worst-case induced $\mathcal{L}_2$-gain of $0.124$ from $\tilde{n}$ to $\tilde{u}$ as in Table~\ref{table:1}. Since, the robust controller explicitly accounts for model uncertainty, it has approximately $37.8\%$ lower worst-case induced $\mathcal{L}_2$-gain from $\tilde{d}_{in}$ to $\tilde{u}$ compared to the nominal controller. Similarly, the robust controller performs better in terms of bounding the Euclidean outputs as compared to the nominal controller at $\beta = 0.8$. Overall, the disturbance rejection property for the nominal controller is degraded more from Table~\ref{table:1} to Table~\ref{table:2} as compared to the robust controller. This observation is consistent with known frequency domain insights for infinite horizon LTI systems, as the high frequency noise rejection properties are typically less impacted by model uncertainties than the low frequency disturbance rejection properties.
\begin{figure}[h]
	\centering
	\includegraphics[width=0.60\linewidth]{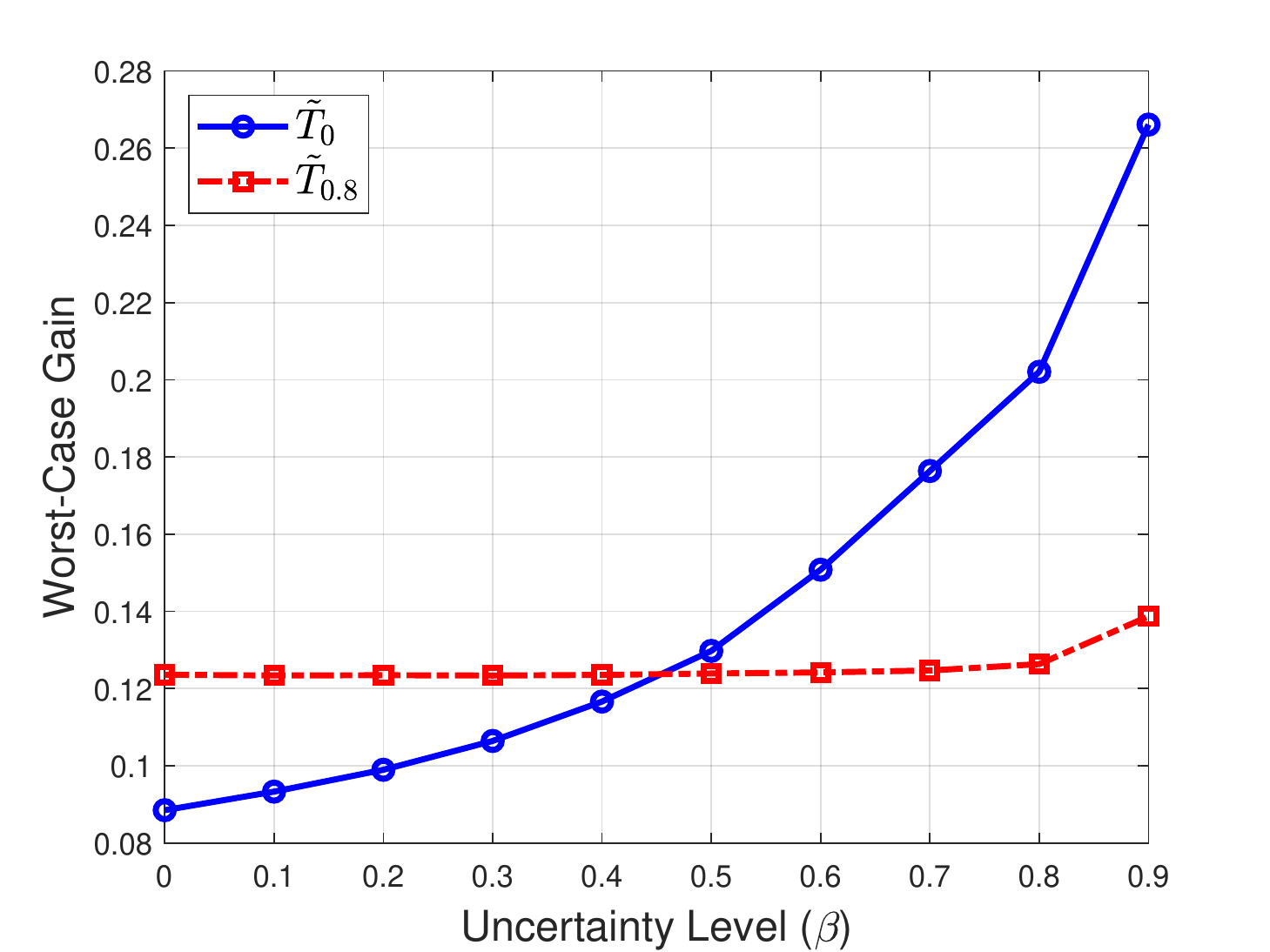}
	\caption{Worst-Case Gain Comparison}
	\label{fig:CompareBothControllers}
\end{figure}
\begin{table}[h]
	\centering
	\begin{tabular}{ |c|c|c|c| }	
		\hline
		\commentb{$\tilde{T}_{0}$} / \commentr{$\tilde{T}_{0.8}$} & \textbf{Disturbance $\tilde{d}_{in}$}
		& \textbf{Measurement noise $\tilde{n}$}
		& \textbf{Generalized disturbance $\tilde{d}$} \\
		\hline
		\textbf{Euclidean output $\tilde{e}_E(T)$} & \commentb{0.066} / \commentr{0.077} & \commentb{0.050} / \commentr{0.047} &  \commentb{0.082} /  \commentr{0.081}\\ \hline
		\textbf{Control effort $\tilde{u}$} & \commentb{0.085} / \commentr{0.087} & \commentb{0.088} / \commentr{0.124}	&  \commentb{0.089} / \commentr{0.124} \\ \hline
		\textbf{Generalized error $\tilde{e}$} & \commentb{0.086} / \commentr{0.097} 
		& \commentb{0.088} / \commentr{0.124} 
		& \commentb{0.089} / \commentr{0.124}\\ 
		\hline
	\end{tabular}
	\caption{Induced Gain Upper Bounds for Different Input-Output Pairs (Nominal Analysis, $\beta = 0$)}
	\label{table:1}
\end{table}
\begin{table}[h]
	\centering
	\begin{tabular}{ |c|c|c|c| }	
		\hline
		\commentb{$\tilde{T}_{0}$} / \commentr{$\tilde{T}_{0.8}$}
		& \textbf{Disturbance $\tilde{d}_{in}$}
		& \textbf{Measurement noise $\tilde{n}$}
		& \textbf{Generalized disturbance $\tilde{d}$} \\
		\hline
		\textbf{Euclidean output $\tilde{e}_E(T)$} & \commentb{0.102} / \commentr{0.093} & \commentb{0.068} / \commentr{0.061} & \commentb{0.120} /  \commentr{0.120}\\ \hline
		\textbf{Control effort $\tilde{u}$} & \commentb{0.190} / \commentr{0.118} & \commentb{0.098} / \commentr{0.124}	&  \commentb{0.202} / \commentr{0.126} \\ \hline
		\textbf{Generalized error $\tilde{e}$}
		& \commentb{0.190} / \commentr{0.118} & \commentb{0.106} / \commentr{0.125} & \commentb{0.202} / \commentr{0.126}\\ 
		\hline
	\end{tabular}
	\caption{Worst-Case Gain Upper Bounds for Different Input-Output Pairs (Robust Analysis, $\beta = 0.8$)}
	\label{table:2}
\end{table}

As noted earlier, the primary design goal was to tightly bound the states at the final
time $T = 5$ seconds. To study this further consider the impact of
link joint disturbance $d_{in}$ on the Euclidean output $e_E$. Let $G$ denote the
unweighted plant which has the same inputs/outputs as the weighted
plant $\tilde{G}$ but with all weights set to identity. Further, let
the respective uncertain interconnection using $G$ be denoted as
$T_{0}:= \mathcal{F}_u(\mathcal{F}_l(G,K_0),\Delta)$ and
$T_{0.8}:= \mathcal{F}_u(\mathcal{F}_l(G,K_{0.8}),\Delta)$. 
Nominal analysis performed for both the $T_{0\,(d_{in}\rightarrow e_E)}$ and $T_{0.8\,(d_{in}\rightarrow e_E)}$ interconnections gives both upper and lower bounds on the nominal performance. The upper bounds are obtained as $0.656$ and $0.766$ respectively, which are shown as blue and red disk in \figref{fig:L2toESyn_NomAnalysis} at the final time. The corresponding lower bounds are obtained as $0.648$ and $0.763$. The worst-case disturbance $\|d_{in}\|_{2,[0,T]} \leq 0.5$ for both interconnections are computed by solving the two point boundary value problem as presented in \cite{iannelli2019construction}. These specific bad disturbances (\figref{fig:wcdist_nom}) pushes the state trajectory (dashed line) as far as the computed lower bound in the LTV simulation. 
\begin{figure}[h]
	\centering
	\begin{minipage}{.5\textwidth}
		\centering
		\includegraphics[width=\linewidth]{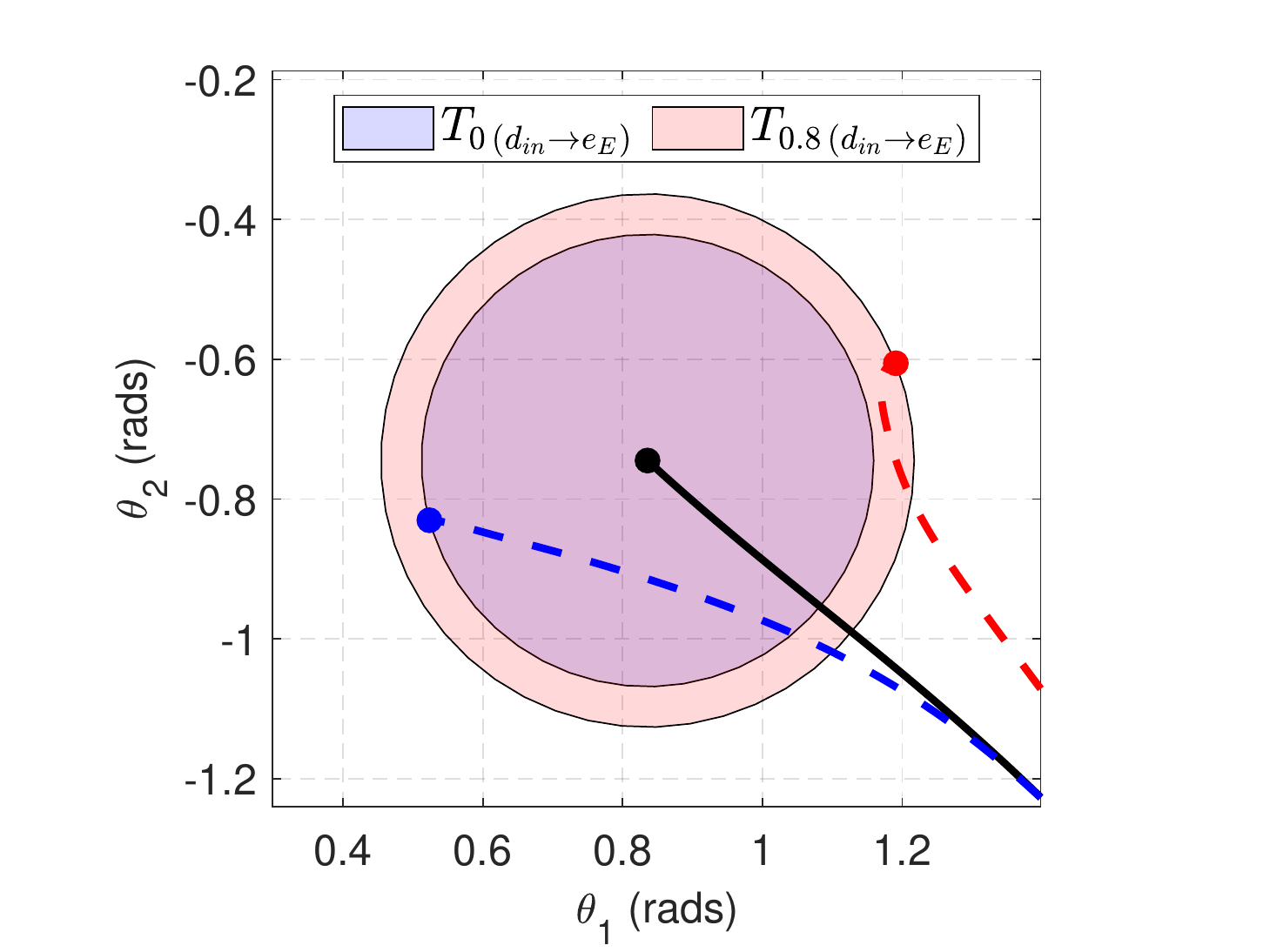}
		\caption{Nominal Analysis $(\beta = 0)$}
		\label{fig:L2toESyn_NomAnalysis}
	\end{minipage}%
	\begin{minipage}{.5\textwidth}
		\centering
		\includegraphics[width=\linewidth]{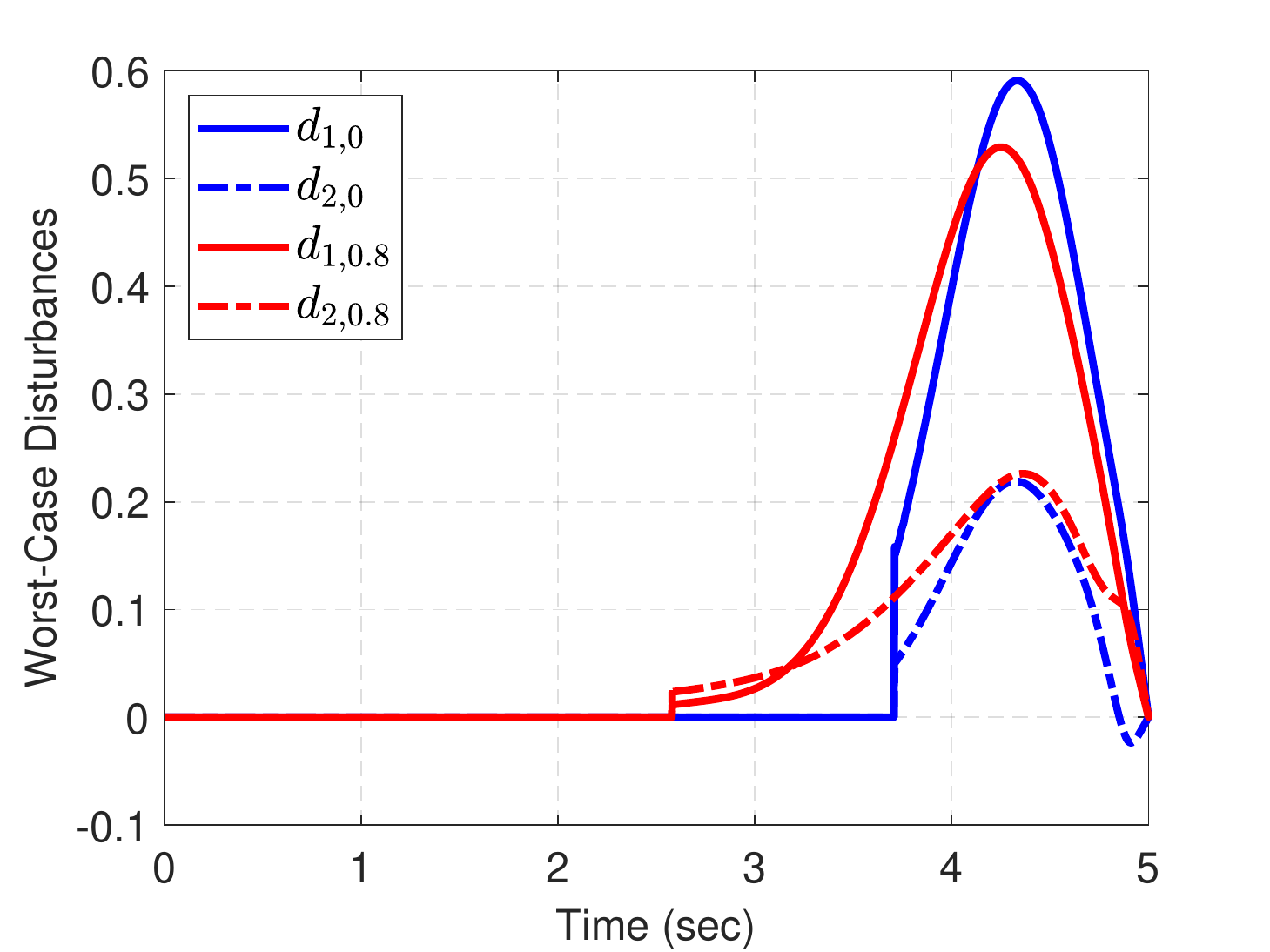}
		\caption{Worst-Case Disturbances}
		\label{fig:wcdist_nom}		
	\end{minipage}
\end{figure}
\begin{figure}[h]
	\centering
	\begin{minipage}{.5\textwidth}
		\centering
		\includegraphics[width=\linewidth]{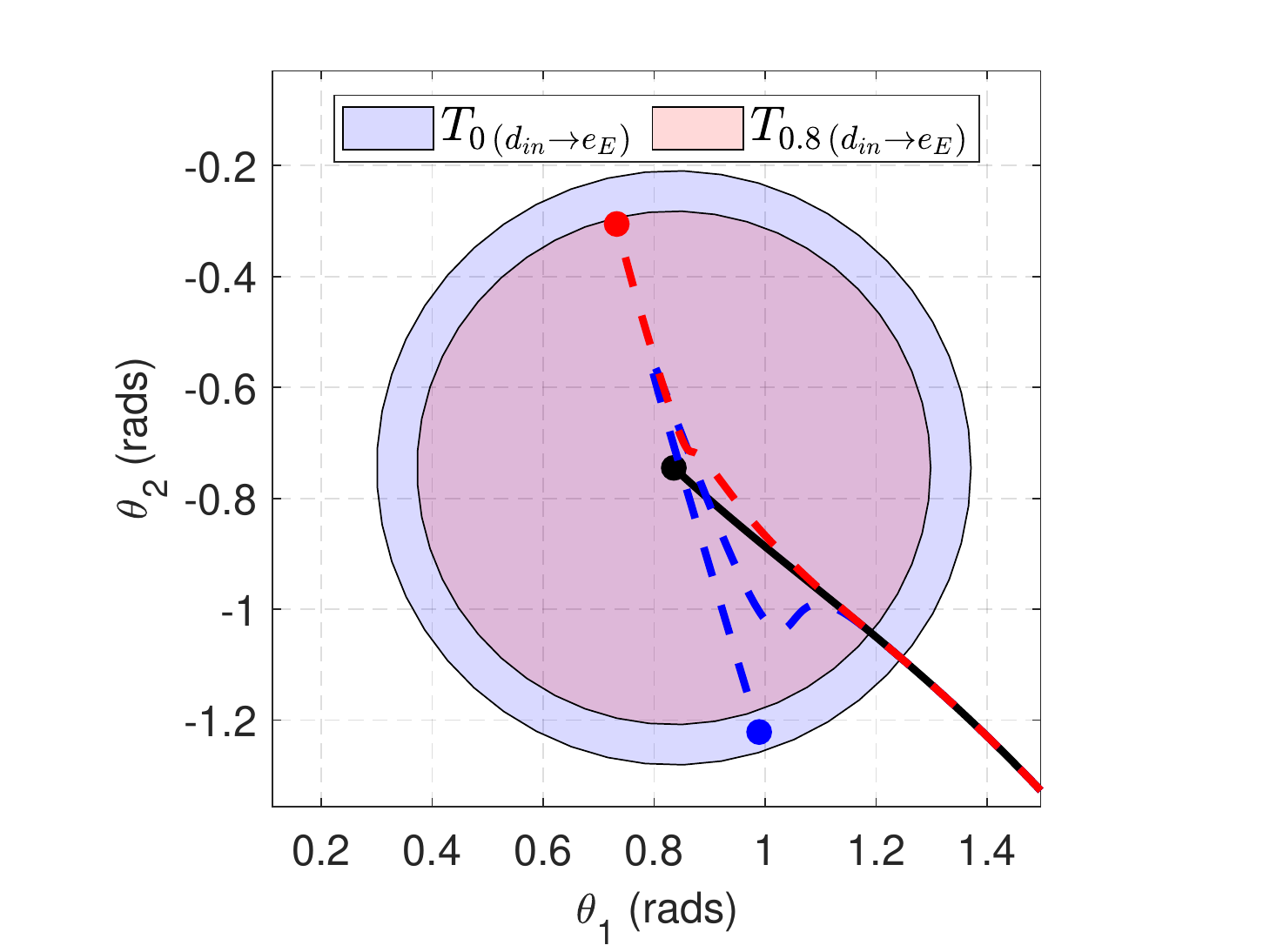}
		\caption{Robust Analysis $(\beta = 0.8)$}
		\label{fig:L2toESyn_RobAnalysis}
	\end{minipage}%
	\begin{minipage}{.5\textwidth}
		\centering
		\includegraphics[width=\linewidth]{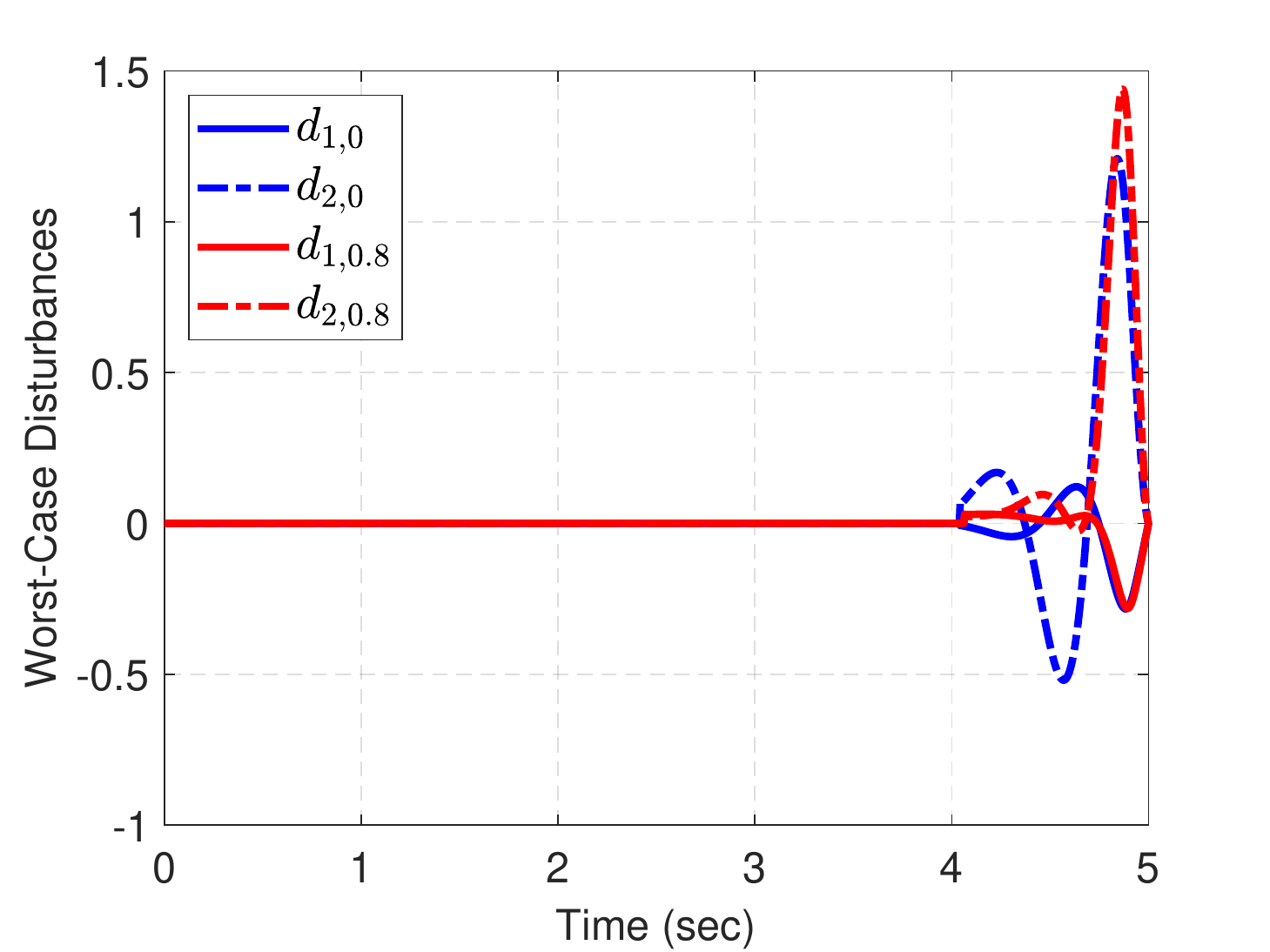}
		\caption{Worst-Case Disturbances}
		\label{fig:wcdist_rob}
	\end{minipage}
\end{figure}

A worst-case terminal Euclidean norm bound is computed for both the interconnections at the
uncertainty level $\beta=0.8$. The corresponding upper bound using
the algorithm in \cite{seiler2019finite} was obtained as $1.02$ and
$0.93$, respectively. This shows approximately a $8.82\%$ reduction
in Euclidean norm bound. As a graphical illustration, these bounds are depicted in \figref{fig:L2toESyn_RobAnalysis} as a disk at the final time $T=5$ seconds. The bound accounts for all the disturbances $d_{in}$ that satisfy $\|d_{in}\|_{2,[0,T]} \leq 0.5$ and all the LTI uncertainties $\Delta$ with norm bound $\beta = 0.8$.

To obtain a reasonable lower bound on the worst-case gain, first $100$ uncertainties are sampled randomly as first order LTI systems with at most size $0.8$. Then, uncertainty block $\Delta$ was replaced with each of sampled uncertainties and nominal LTV analysis was performed from $d_{in}$ to $e_E$ for both the interconnections. Worst-case uncertainties are then obtained after maximizing performance over the sample space. Let, the specific bad perturbation that yields to the poor performance for both $T_{0\,(d_{in}\rightarrow e_E)}$ and $T_{0.8\,(d_{in}\rightarrow e_E)}$ be denoted as $\Delta_{wc_1}$ and $\Delta_{wc_2}$ respectively.
\begin{align*}
	\Delta_{wc_1} = \frac{-0.8\,s + 12.18}{s + 15.23}, \hspace{0.2in} \Delta_{wc_2} = \frac{-0.8\,s + 25.89}{s+32.36}
\end{align*}

The worst-case gain lower bound corresponding to these perturbations
are obtained as $1.001$ and $0.903$ respectively. It is evident that a
combination of worst-case disturbance (scaled to have size $0.5$) and uncertainty (of size $0.8$) pushes the
states of the closed loop system (dashed line) as far as the lower bound of the
worst-case gain. Overall, these simple comparison results show a typical robustness and performance 
trade-off. The nominal controller performs best at no uncertainty whereas the robust controller 
performs better at modeled uncertainty level.

\section{Conclusions \& Future Work}
\label{sec:Conclusions}
This paper proposed an iterative algorithm to design an output-feedback controller that bounds the worst-case gain of an uncertain LTV system on a finite horizon. Similar design can also be done in a state-feedback formulation. The performance was specified using both an induced $\mathcal{L}_2$ and terminal Euclidean norm penalty on the output. Time-domain dynamic IQCs were used to describe the input-output behavior of the uncertainty. The effectiveness of proposed approach was demonstrated using the two-link robot arm example.

This paper opens up few new directions for further research. First, we used a block diagonal assumption for the IQC filter and decision variables (Assumption~\ref{asm:blkdiag}). Future work will consider relaxing this assumption to include full block IQC multipliers in the robust synthesis. Second, note that the proposed method also allows time-varying uncertainty and performance weights. This is useful for many applications such as in many launch scenarios, where the uncertainty or performance requirements are not evenly spread out across the time horizon. Future work in this area is required to exploit the full potential of this method. Moreover, we recognize that proposed method is computationally expensive. Future research is needed to speed up the numerical computations for such finite horizon analysis and synthesis.

\section*{Acknowledgments}                            
This work was gratefully supported by the US ONR grant N00014-18-1-2209. We would like to thank our collaborators Prof. Andrew Packard, Prof. Murat Arcak, and graduate students Kate Schweidel, Emmanuel Sin, Alex Devonport and Galaxy Yin for valuable discussions. We also thank Dr. Douglas Philbrick from U.S. Naval Air Warfare Center Weapons Division, China Lake for helpful comments. Finally, we thank Prof. Harald Pfifer for pointing out the technical issue with Lemma 2 in earlier draft of this paper.

\appendix
\section{Generic Quadratic Cost}
\label{sec:quadcost}

This paper considers an induced norm $\|H\|_{[0,T]}$ (defined by Equation~\ref{eq:norm}) as a performance metric whereas a generic quadratic cost is considered in~\cite{seiler2019finite}. This appendix describes the equivalence between these two formulations. First, consider the finite-horizon linear quadratic optimal control problem as follows:
\begin{align}
	J^*(\gamma):= &\sup_{0 \ne d \in \mathcal{L}_2[0,T]} x(T)^\top F x(T) + \int_0^T \bmtx x(t) \\ d(t) \emtx^\top \bmtx Q(t) & S(t)\\ S(t)^\top & R(t,\gamma) \emtx \bmtx x(t) \\ d(t) \emtx dt\nonumber\\
	&\hspace{0.25in}\mbox{s.t.  } \dot{x}(t) = A(t)\, x(t) + B(t)\, d(t) \mbox{ and } x(0) = 0.
	\label{eq:LQoptprob}
\end{align}
where $Q : [0,T] \rightarrow \Sm^{n_x}$, $S: [0,T] \rightarrow \R^{n_x \times n_d}$, $R : [0,T] \oplus \R^+ \rightarrow \Sm^{n_d}$ and $F \in \R^{n_x\times n_x}$. We assume $Q(t) \succeq 0$, for all $t\in[0,T]$ and $F\succeq 0$. Moreover, we assume a form $R(t,\gamma) = R_0(t) - \gamma^2 I_{n_d}$, where $\gamma > 0$, $R_0(t) \succeq 0$ and $R(t,\gamma) \prec 0$, for all $t\in[0,T]$. There are two directions to the equivalence. First, assume a system $H$ is given as defined by Equation~\eqref{eq:nLTV} and \eqref{eq:LTVoutput}. Note that the induced norm $\|H\|_{[0,T]}$ is defined by the state-space matrices $(A,B,C_I,C_E,D_I)$. Define $(Q,S,R,F)$ as in Theorem~\ref{thm:nominalperf}. Then for any $\gamma>0$, $\|H\|_{[0,T]} < \gamma$ if and only if $J^*(\gamma)<0$. This is shown in Section $2$ of~\cite{seiler2019finite}. Conversely, assume the generic quadratic cost defined by Equation~\eqref{eq:LQoptprob} is given with cost matrices $(Q,S,R,F)$ satisfying the assumptions above. If we further assume that $Q(t)-S(t) R_0(t)^{-1} S(t)^\top \succ 0$ then we can perform the following factorization:
\begin{align}
	\bmtx Q(t) & S(t)\\ S(t)^\top & R_0(t) \emtx = \bmtx C_I(t)^\top \\ D_I(t)^\top \emtx \bmtx C_I(t) & D_I(t) \emtx
\end{align} 
In addition, define $C_E := F^{\frac{1}{2}}$. Then the generic quadratic cost is re-written as:
\begin{align}
	J^*(\gamma) = &\sup_{0\ne d \in \mathcal{L}_2[0,T]} e_E(T)^\top e_E(T) + \int_0^T e_I(t)^\top e_I(t) \,dt - \gamma^2 \int_0^T d(t)^\top \, d(t) \, dt \nonumber\\
	& \hspace{0.25in} \mbox{s.t. Equation}~\eqref{eq:nLTV},~\eqref{eq:LTVoutput} \mbox{ and } x(0) = 0. 
\end{align}
This cost satisfies $J^*(\gamma)<0$ if and only if $\|H\|_{[0,T]} < \gamma$.

\section{Finite Horizon Factorization}
\label{sec:spectralfact}

For infinite horizon LTI systems, spectral factorization results are found in standard robust control textbooks \cite{zhou1996robust,dullerud2013course,francis1987course}. The following lemma provides a time-varying finite horizon generalization of this result.
\begin{lem}
	\label{lem:spectralfact}
	Consider an LTV system $\Psi:\mathcal{L}_{2}^{n_d}[0,T] \rightarrow \mathcal{L}_{2}^{n_e}[0,T]$ be given with state-space realization as follows:
	\begin{align}
		\label{eq:sLTV}
		\begin{split} 
			\dot{x}(t) &= A(t)\, x(t) + B(t)\, d(t)\\
			e(t) & = C(t) \, x(t) + D(t) \, d(t)
		\end{split}
	\end{align}
	with $x\in\R^{n_x}$, $e\in\R^{n_e}$, $d\in\R^{n_d}$ and $D(t)$ is full column rank $\forall t\in[0,T]$. Let $M:[0,T]\rightarrow\Sm^{n_e}$ be a given piecewise continuous matrix valued function with $M(t) \succ 0, \forall t\in[0,T]$. Let $Q : [0,T] \rightarrow \Sm^{n_x}$,
	$S : [0,T] \rightarrow \R^{n_x \times n_d}$,
	$R : [0,T] \rightarrow \Sm^{n_d}$ be defined as follows.
	\begin{align}
		Q := C^\top M C, \hspace{0.2in} S := C^\top M D, \hspace{0.2in} R := D^\top M D
	\end{align}
	with $R(t)\succ 0, \forall t\in[0,T]$. The following statements hold.
	\begin{enumerate}
		\item There exist a differentiable function $X:[0,T]\rightarrow\Sm^{n}$ such that $X(T) = 0$ and 
		\begin{align}			
			\dot{X} + A^\top X + XA + Q - (XB + S)R^{-1}(XB+ S)^\top= 0
			\label{eq:spectralfactDRE}
		\end{align}	
		\item $\Phi:=\Psi^\sim M \Psi$ has a finite horizon factorization $\Phi = U^\sim U$ where $U$ is square invertible LTV system defined on $[0,T]$ with the following state-space realization:
		\begin{align}
			U = \left[
			\begin{array}{c|c c}
				A \, & B\\
				\hline
				W^{-T}(B^\top X+S^\top)& W
			\end{array}
			\right]
		\end{align}
		where $R(t) = W(t)^\top W(t)$, $\forall t\in[0,T]$.
	\end{enumerate}
\end{lem}
\begin{proof}
	Since $R(t) \succ 0$ and by Schur complement lemma $Q(t) - S(t)
	R(t)^{-1} S(t)^\top \succ 0, \forall t\in[0,T]$, the RDE does not have a finite escape time and thus always have a bounded unique solution regardless of the boundary condition (Corollary 2.3 of \cite{molinari1975nonnegativity}, Theorem $8$ in \cite{kuvcera1973review}). Further, it can be verified  that the time-varying state-space realization of $\Psi^\sim M \Psi$ is related to that of a system $U^\sim U$ by a similarity transformation matrix $\bsmtx I & 0\\X(t) & I\esmtx$. 
\end{proof}

\section{Proof of Lemma 1}
\label{sec:proof1}

\textbf{Lemma 1. } 	Let $\epsilon > 0$, $\gamma>0$, $M_v(t)\succ 0$, $M_w(t)\succ 0$ and a differentiable function $P:[0,T]\rightarrow\Sm^{n}$ such that $P(T)\geq F$ be given with the choice of $(Q,S,R,F)$ as in Equation~\eqref{eq:QSRFRobPerf}. The following statements are equivalent:
\begin{enumerate}
	\item $DLMI_{Rob}(P,M,\gamma^2,t) \leq -\epsilon I, \;\hspace{0.1in}\; \forall t\in[0,T].$
	\item $\|N_{scl}\|_{[0,T]} \leq 1 - \hat{\epsilon}$,\hspace{0.1in} for some $\hat{\epsilon}$.
\end{enumerate}
\begin{proof}
	$\mathbf{(1\Rightarrow 2)}$ This proof is presented in two parts. First, we show that $DLMI_{Rob}(P,M,\gamma^2,t) \leq -\epsilon I, \; \forall t\in[0,T]$ can equivalently be written as a dissipation inequality with only single valid IQC. Second, the state-space realization of the extended system $N_{ext}$ and the scaled system $N_{scl}$ are indeed the same, which allow us to rewrite the robust performance DLMI as a nominal performance DLMI for $N_{scl}$. Integrating the related dissipation inequality completes the proof.\\ 
	
	\noindent \textbf{Part (i):} Define a storage function $V(x,t):=x^\top P(t)x$. Left and right multiply the DLMI \eqref{eq:DLMI} by $[x^\top, w^\top, d^\top]$ and its transpose to show that $V$ satisfies the following dissipation inequality for all $t\in [0,T]$:
	\begin{align}
		\label{eq:DI2}
		\dot{V} + 
		\bmtx x \\ \bsmtx w \\ d \esmtx \emtx^\top
		\bmtx Q & S \\ S^\top & R \emtx 
		\bmtx x \\ \bsmtx w \\ d \esmtx \emtx
		+ z^\top M z
		\le -\epsilon \, d^\top d.
	\end{align}
	where $x = \bsmtx x_N \\ x_\psi\esmtx \in \R^n$ is the state of extended system as shown in \figref{fig:IQCAugment}. Consider the outputs of the IQC filter $\Psi= \bsmtx \Psi_v & 0 \\ 0 & \Psi_w\esmtx$ be partitioned as $z~:=~\bsmtx z_v\\z_w\esmtx$.  Let $\Psi_v$ have the following state-space representation with state $x_v$, input $v$, and output $z_v$:	
	\begin{align}
		\begin{split}
			\dot{x}_v(t) = A_1(t)\, x_v(t) + B_1(t)\, v(t)\\
			z_v(t) = C_1(t) \, x_v(t) + D_1(t)\, v(t)
		\end{split}	
	\end{align}		
	A similar time-varying state-space expression also holds for $\Psi_w$ with matrices $(A_2,B_2,C_2,D_2)$, state $x_w$, input $w$, and output $z_w$. Thus the term $z^\top M z$ in \eqref{eq:DI2} can be expressed as:
	\begin{align}
		\label{eq:zTMz}
		\begin{split}
			z^\top M z &= z_v^\top M_v z_v - z_w^\top M_wz_w\\
			&= \bmtx x_v\\v\emtx^\top \bmtx C_1^\top \\D_1^\top \emtx  M_v \bmtx C_1 & D_1 \emtx \bmtx x_v\\v\emtx
			- \bmtx x_w\\w\emtx^\top \bmtx C_2^\top \\D_2^\top \emtx M_w \bmtx C_2 & D_2 \emtx \bmtx x_w\\w\emtx
		\end{split}
	\end{align}
	First, consider only the terms involving $v$ and define the quadratic storage matrices as:
	\begin{align}
		\label{eq:QuadCost}
		\bmtx Q_v & S_v\\ S_v^\top & R_v \emtx &:= \bmtx C_1^\top \\ D_1^\top \emtx  M_v \bmtx C_1 & D_1 \emtx
	\end{align}	
	By Lemma~\ref{lem:spectralfact} in Appendix~\ref{sec:spectralfact}, the condition $M_v(t)\succ 0$, $\forall t \in [0,T]$ implies that there exists $X_v:[0,T]\rightarrow\Sm^{n_{x_v}}$ such that:
	\begin{align}
		\label{eq:DRE}
		\begin{split}
			\dot{X}_v + A_1^\top X_v + X_vA_1 + Q_v - (X_vB_1 + S_v)R_v^{-1}(X_vB_1 + S_v)^\top = 0, \hspace{0.2in} X_v(T) = 0
		\end{split}	
	\end{align}
	Moreover, Lemma~\ref{lem:spectralfact} in Appendix~\ref{sec:spectralfact} also implies that there exists a spectral factor $U_v$ with a state-space realization as $(A_1,B_1,\tilde{C}_1,\tilde{D}_1)$ with $\tilde{C}_1 := W_v^{-T}(B_1^\top X_v+S_v^\top)$,  $\tilde{D}_1 := W_v$ and $R_v = W_v^\top W_v$. Note that $x_v$ is the state and $\tilde{v}$ is the output of the spectral factor $U_v$. The RDE \eqref{eq:DRE} can be written in terms of the state matrices of $U_v$ as: 
	\begin{align}
		Q_v = -\dot{X}_v -A_1^\top X_v - X_v A_1 +  \tilde{C}_1^\top \tilde{C}_1
	\end{align}
	Substitute above $Q_v$ and $S_v^\top = \tilde{D}_1^\top\tilde{C}_1 - B_1^\top X_v$ in $\eqref{eq:QuadCost}$ to obtain the following expression:	
	\begin{align}
		z_v^\top M_vz_v = & -x_v\dot{X}_v x_v -(A_1 x_v + B_1 v)^\top X_v x_v - x_v^\top X_v (A_1 x_v + B_1 v) + (\tilde{C}_1 x_v + \tilde{D}_1 v)^\top (\tilde{C}_1 x_v + \tilde{D}_1 v)
	\end{align}
	This can be simplified to the following expression:	
	\begin{align}	
		z_v^\top M_vz_v = -x_v^\top \dot{X}_v x_v - \dot{x}_v^\top X_v x_v -x_v^\top X_v \dot{x}_v +  \tilde{v}^\top\tilde{v}= -\frac{d}{dt} \left(x_v^\top X_v x_v\right) +  \tilde{v}^\top\tilde{v}
		\label{eq:vPiv}
	\end{align}
	Similarly, with $M_w(t) \succ 0$, $\forall t \in [0,T]$ the spectral factor $U_w$ can be obtained with a state-space realization $(A_2,B_2,\tilde{C}_2,\tilde{D}_2)$, states $x_w$, and outputs $\tilde{w}$. The following expression holds:
	\begin{align}
		\label{eq:wPiw}
		z_w^\top M_wz_w = -\frac{d}{dt} \left(x_w^\top X_w x_w\right) + \tilde{w}^\top\tilde{w}
	\end{align}
	where $X_w:[0,T]\rightarrow\Sm^{n_{x_w}}$ is a solution to a related RDE with respective quadratic storage matrices and the boundary condition $X_w(T) = 0$. Subtract Equation~\eqref{eq:wPiw} from \eqref{eq:vPiv} to get the left hand side of the Equation~\eqref{eq:zTMz} as follows:
	\begin{align}
		\label{eq:zTMzPi}
		\begin{split}
			z^\top Mz &= -\frac{d}{dt} \left(x_\psi^\top X x_\psi \right) + \tilde{v}^\top\tilde{v} - \tilde{w}^\top\tilde{w}
		\end{split}
	\end{align}		
	where $x_\psi = \bsmtx x_v\\ x_w\esmtx$, $X(t) := \bsmtx X_v(t) & 0 \\ 0 & -X_w(t) \esmtx$ and $X(T) = 0$. Let the modified matrix $\tilde{P}(t)$ be defined as follows:
	\begin{align}
		\label{eq:Ptilde}
		\tilde{P}(t) := P(t) - \bmtx 0 & 0\\0 & X(t)\emtx
	\end{align}
	This yields a modified storage function $\tilde{V}(x,t):= x^\top \tilde{P}(t)x$. The modified storage function has the form:
	\begin{align}
		\label{eq:newstorage}
		\tilde{V}(x,t) = V(x,t) - x_v^\top X_vx_v + x_w^\top X_wx_w
	\end{align}
	where the second and third term can be interpreted as hidden energy stored in the IQC multiplier. With modified storage function $\tilde{V}$ the dissipation inequality \eqref{eq:DI2} can be recast as,
	\begin{align}
		\label{eq:DI6}
		\begin{split}
			\dot{\tilde{V}} + \bmtx x \\ \bsmtx w \\ d \esmtx \emtx^\top
			\bmtx Q & S \\ S^\top & R \emtx 
			\bmtx x \\ \bsmtx w \\ d \esmtx \emtx + \bmtx \tilde{v} \\ \tilde{w}\emtx^\top J_{n_v,n_w} \bmtx \tilde{v} \\ \tilde{w}\emtx
			\le -\epsilon \, d^\top d
		\end{split}
	\end{align}
	This dissipation inequality is equivalent to \eqref{eq:DI2} with a single IQC $\mathcal{I}(U,J_{n_v,n_w})$ where
	$U:= \bsmtx U_v & 0 \\0 & U_w \esmtx$. Next, We show that $\mathcal{I}(U,J_{n_v,n_w})$ is a valid time-domain IQC. To see this, define $V_{\psi}(x_\psi(t),t):= x_\psi(t)^\top X(t) x_\psi(t)$, $\tilde{z}(t) := \bsmtx \tilde{v}(t)\\ \tilde{w}(t)\esmtx$ and integrate Equation~\eqref{eq:zTMzPi} both sides from $0$ to $T$ to obtain:
	\begin{align}
		\int_0^T z^\top Mz \; dt &= - V_{\psi}(x_\psi(T),T)  + V_{\psi}(x_\psi(0),0) + \int_0^T \tilde{z}^\top J_{n_v,n_w} \tilde{z} \; dt
	\end{align}
	Note that $V_{\psi}(x_\psi(0),0) = 0$ because $x_\psi(0)=0$ and $V_{\psi}(x_\psi(T),T)= x_\psi(T)^\top X(T) x_\psi(T) = 0$ due to the boundary condition $X(T)=0_{n_\psi}$ of the time-varying factorization RDE. Thus if $\int_0^T z^\top Mz \; dt \geq 0$ then we have $\int_0^T \tilde{z}^\top J_{n_v,n_w} \tilde{z} \; dt \geq 0$. Finally, note that $\tilde{P}(t)$ satisfies the same boundary condition as $P(t)$ i.e. $\tilde{P}(T) \geq F$ because of the boundary condition $X(T) = 0$. Thus, $\tilde{V}(x,t)$ is a valid storage function $\forall t\in[0,T]$.\\ 
	
	\noindent \textbf{Part (ii):} Let the extended system of $N$ with spectral factor $U$ be written in partitioned form as:
	\begin{align}
		\label{eq:extsys2}
		\bmtx \dot{x}\\\hline \tilde{v}\\ \tilde{w}\\ e_{I} \\ e_E\emtx =
		\left[
		\begin{array}{c|c c}
			\mathcal{A} & \mathcal{B}_w & \mathcal{B}_d\\
			\hline
			\mathcal{C}_{\tilde{v}} & \mathcal{D}_{\tilde{v}w} & \mathcal{D}_{\tilde{v}d}\\
			\mathcal{C}_{\tilde{w}} & \mathcal{D}_{\tilde{w}w} & 0 \\
			\mathcal{C}_{I} & \mathcal{D}_{Iw} & \mathcal{D}_{Id}\\
			\mathcal{C}_E & 0 & 0
		\end{array}
		\right]
		\bmtx x \\\hline w\\ d\emtx
	\end{align}
	where $x = \bsmtx x_N \\ x_v\\x_w\esmtx \in \R^n$ and state-space matrices: 
	\begin{align*}
		&\mathcal{A} := \bmtx A_N & 0 & 0\\  B_{1} C_{v} & A_{1} & 0\\ 0 & 0& A_2\emtx,
		\mathcal{B}_w:= \bmtx B_{w} \\ B_{1} D_{vw}\\B_2 \emtx,\mathcal{B}_d:= \bmtx B_{d} \\ B_{1} D_{vd}\\0 \emtx\\
		&\mathcal{C}_{\tilde{v}} := \bmtx \tilde{D}_{1} C_{v} & \tilde{C}_{1} & 0\emtx,
		\mathcal{C}_I := \bmtx C_{I} & 0 & 0\emtx,
		\mathcal{C}_{\tilde{w}} := \bmtx 0 & 0 &\tilde{C}_2 \emtx\\
		&\mathcal{C}_E := \bmtx C_E & 0 & 0\emtx,
		\mathcal{D}_{\tilde{v}w} := \tilde{D}_1D_{vw}, \mathcal{D}_{\tilde{v}d} := \tilde{D}_1D_{vd}\\
		&\mathcal{D}_{\tilde{w}w}:=\tilde{D}_2,\mathcal{D}_{Iw} := D_{Iw}, \mathcal{D}_{Id} = D_{Id}
	\end{align*}
	Using the choice of $(Q,S,R)$ from Equation~\eqref{eq:QSRFRobPerf}, the following partitioned DLMI is equivalent to the dissipation inequality~\eqref{eq:DI6} for the state-space realization of~\eqref{eq:extsys2}.
	\begin{align}
		\label{eq:DLMITil}
		\bmtx
		\dot{\tilde{P}}+\mathcal{A}^\top\tilde{P}+\tilde{P}\mathcal{A} & \tilde{P}\mathcal{B}_w & \tilde{P}\mathcal{B}_d\\
		\mathcal{B}_w^\top\tilde{P} & 0_{n_w} & 0\\
		\mathcal{B}_d^\top\tilde{P} & 0 & -\gamma^2I_{n_d}
		\emtx +
		\bmtx \mathcal{C}_I^\top\\ \mathcal{D}_{Iw}^\top \\ \mathcal{D}_{Id}^\top\emtx \bmtx \mathcal{C}_I & \mathcal{D}_{Iw} & \mathcal{D}_{Id}\emtx 
		+ \bmtx \mathcal{C}_{\tilde{v}}^\top & \mathcal{C}_{\tilde{w}}^\top\\
		\mathcal{D}_{\tilde{v}w}^\top& \mathcal{D}_{\tilde{w}w}^\top\\
		\mathcal{D}_{\tilde{v}d}& 0\emtx
		J_{n_v,n_w} \bmtx
		\mathcal{C}_{\tilde{v}} & \mathcal{D}_{\tilde{v}w}& \mathcal{D}_{\tilde{v}d}\\
		\mathcal{C}_{\tilde{w}} & \mathcal{D}_{\tilde{w}w}& 0
		\emtx\leq -\epsilon I
	\end{align}
	The condition  $M_w(t) \succ 0, \forall t\in[0,T]$ is sufficient to ensure that $\mathcal{D}_{\tilde{w}w} := \tilde{D}_2$ is nonsingular. The output equation for $w$ can be written as: $w=\mathcal{D}_{\tilde{w}w}^{-1}(\tilde{w} - \mathcal{C}_{\tilde{w}}x)$. Use this relation to substitute for $w$ in Equation~\eqref{eq:extsys2}. This gives the following scaled system:   
	\begin{align}
		\label{eq:ScaledSys}
		\bmtx \dot{x}\\\hline \tilde{v}\\ e_{I} \\ e_E\emtx =
		\left[
		\begin{array}{c|c c}
			\mathcal{A} & \mathcal{B}_w & \mathcal{B}_d\\
			\hline
			\mathcal{C}_{\tilde{v}}& \mathcal{D}_{\tilde{v}w} & \mathcal{D}_{\tilde{v}d} \\
			\mathcal{C}_{I} & \mathcal{D}_{Iw} & \mathcal{D}_{Id}\\
			\mathcal{C}_E & 0 & 0
		\end{array}
		\right] L
		\bmtx x \\\hline \tilde{w}\\ d\emtx
	\end{align}
	where the nonsingular time-varying matrix $L$ is defined as:
	\begin{align}
		L := \bmtx I_n & 0 & 0\\ 
		-\mathcal{D}_{\tilde{w}w}^{-1}\mathcal{C}_{\tilde{w}} & \mathcal{D}_{\tilde{w}w}^{-1}& 0\\
		0 &0 & I_{n_d}\emtx
	\end{align}
	Equation~\eqref{eq:ScaledSys} can be rewritten as follows:
	\begin{align}
		\label{eq:ScaledSys2}
		\bmtx \dot{x}\\\hline \tilde{v}\\ e_{I} \\ e_E\emtx =
		\left[
		\begin{array}{c|c c}
			\tilde{\mathcal{A}} & \mathcal{B}_{\tilde{w}} & \mathcal{B}_d\\
			\hline
			\tilde{\mathcal{C}}_{\tilde{v}} & \mathcal{D}_{\tilde{v}\tilde{w}} & \mathcal{D}_{\tilde{v}d} \\
			\tilde{\mathcal{C}}_{I} & \mathcal{D}_{I\tilde{w}} & \mathcal{D}_{Id}\\
			\mathcal{C}_E & 0 & 0
		\end{array}
		\right] 
		\bmtx x \\\hline \tilde{w}\\ d\emtx
	\end{align}
	where the updated state-space matrices are defined as:
	\begin{align*}
		\begin{split}
			\tilde{\mathcal{A}} &:= \bmtx A_N & 0 & -B_w\tilde{D}_2^{-1}\tilde{C}_2 \\
			B_1 C_v & A_1 & -B_1D_{vw}\tilde{D}_2^{-1}\tilde{C}_2\\
			0 & 0 & A_2-B_2\tilde{D}_2^{-1}\tilde{C}_2\emtx,\, \mathcal{B}_{\tilde{w}} := \bmtx B_{w}\tilde{D}_2^{-1}\\ B_{1} D_{vw}\tilde{D}_2^{-1}\\B_2\tilde{D}_2^{-1} \emtx\\
			\tilde{\mathcal{C}}_{\tilde{v}} &= \bmtx \tilde{D}_1C_v & \tilde{C}_1 & -\tilde{D}_1 D_{vw}\tilde{D}_2^{-1}\tilde{C}_2 \emtx,\,\tilde{\mathcal{C}}_{I} := \bmtx C_I & 0 & -D_{Iw}\tilde{D}_2^{-1}\tilde{C}_2\emtx\\
			\mathcal{D}_{\tilde{v}\tilde{w}} &:= \tilde{D}_1D_{vw}\tilde{D}_2^{-1},\, \mathcal{D}_{I\tilde{w}} := D_{Iw}\tilde{D}_2^{-1}
		\end{split}
	\end{align*}
	Note that the following state-space matrices of the inverse system of $U_w$ shows up in the above representation. 
	\begin{align}
		U_w^{-1} := \left[\begin{array}{c|c}
			A_2-B_2\tilde{D}_2^{-1}\tilde{C}_2 & B_2\tilde{D}_2^{-1}\\
			\hline
			-\tilde{D}_2^{-1}\tilde{C}_2 & \tilde{D}_2^{-1}
		\end{array}\right]
	\end{align}
	Let scaled signal $\tilde{d} := \gamma \, d\,$ and state-space matrices ($\mathcal{A}_{scl}$, $\mathcal{B}_{scl}$, $\mathcal{C}_{scl}$, $\mathcal{D}_{scl}$) be defined as follows:
	\begin{align}
		\label{eq:ssscaled}
		\begin{split}
			\mathcal{A}_{scl} := \tilde{\mathcal{A}}, \hspace{0.2in} 
			\mathcal{B}_{scl} := \bmtx \mathcal{B}_{\tilde{w}} & \gamma^{-1}\mathcal{B}_d\emtx, \hspace{0.2in} 
			\mathcal{C}_{scl} := \bmtx \tilde{\mathcal{C}}_{\tilde{v}}\\ \tilde{\mathcal{C}}_{I} \emtx, \hspace{0.2in} 
			\mathcal{D}_{scl} := \bmtx\tilde{D}_1D_{vw}\tilde{D}_2^{-1} & \gamma^{-1}\tilde{D}_1D_{vd}\\D_{Iw}\tilde{D}_2^{-1} & \gamma^{-1}D_{Id}\emtx
		\end{split}
	\end{align}
	It is readily verified that, with above definition, the scaled plant $N_{scl}$ has a state-space realization as follows:
	\begin{align}
		\label{eq:ScaledSys3}
		\bmtx \dot{x}\\\hline \bmtx\tilde{v}\\ e_{I} \emtx\\ e_E\emtx =
		\left[
		\begin{array}{c|c c}
			\mathcal{A}_{scl} & \mathcal{B}_{scl}\\
			\hline
			\mathcal{C}_{scl} & \mathcal{D}_{scl} \\
			\mathcal{C}_E & 0 
		\end{array}
		\right]
		\bmtx x \\\hline \bmtx \tilde{w}\\ \tilde{d}\emtx \emtx
	\end{align}
	Perform the congruence transformation by multiplying the DLMI \eqref{eq:DLMITil} on the left/right by $L^\top/L$ to get: 
	\begin{align}
		\begin{split}
			\label{eq:DLMICong}
			\bmtx
			\dot{\tilde{P}}+\tilde{\mathcal{A}}^\top\tilde{P}+\tilde{P}\tilde{\mathcal{A}} & \tilde{P}\mathcal{B}_{\tilde{w}} & \tilde{P}\mathcal{B}_d\\
			\mathcal{B}_{\tilde{w}}^\top\tilde{P} & 0_{n_{\tilde{w}}} & 0\\
			\mathcal{B}_d^\top\tilde{P} & 0 & -\gamma^2I_{n_d}
			\emtx +
			\bmtx \tilde{\mathcal{C}}_I^\top\\ \mathcal{D}_{I\tilde{w}}^\top \\ \mathcal{D}_{Id}^\top\emtx \bmtx \tilde{\mathcal{C}}_I & \mathcal{D}_{I\tilde{w}} & \mathcal{D}_{Id}\emtx 
			+ \bmtx \tilde{\mathcal{C}}_{\tilde{v}}^\top & 0\\
			\mathcal{D}_{\tilde{v}\tilde{w}}^\top& I_{n_{\tilde{w}}}\\
			\mathcal{D}_{\tilde{v}d}^\top& 0\emtx
			J \bmtx
			\tilde{\mathcal{C}}_{\tilde{v}} & \mathcal{D}_{\tilde{v}\tilde{w}}& \mathcal{D}_{\tilde{v}d}\\
			0 & I_{n_{\tilde{w}}}& 0
			\emtx\leq -\epsilon I
		\end{split}
	\end{align}	
	This DLMI can also be written in more compact notation using the state matrices of $N_{scl}$. Multiply inequality \eqref{eq:DLMICong} left and right by $[x^\top, \tilde{w}^\top, d^\top]$ and its transpose to show that $\tilde{V}(x(t),t)=x(t)^\top \tilde{P}(t) x(t)$ satisfies the following dissipation inequality:
	\begin{align}
		\label{eq:DI3}
		\dot{\tilde{V}} + e_I^\top e_I - \gamma^2d^\top d + \tilde{v}^\top\tilde{v} 	-\tilde{w}^\top\tilde{w} \le -\epsilon \, d^\top d
	\end{align}
	Define $\tilde{d} := \gamma \, d$, $\tilde{\epsilon} := \epsilon\,\gamma^{-2}$ and combine the inputs $\tilde{w}$, $\tilde{d}$ together to rewrite the inequality \eqref{eq:DI3} as follows:
	\begin{align}
		\label{eq:DI4}
		\dot{\tilde{V}} + \bmtx \tilde{v} \\ e_I \emtx^\top\bmtx \tilde{v} \\ e_I \emtx - \bmtx \tilde{w} \\ \tilde{d} \emtx^\top\bmtx \tilde{w} \\ \tilde{d} \emtx \le -\tilde{\epsilon}\, \tilde{d}^\top\tilde{d}
	\end{align}
	Integrate over $[0,T]$ to obtain the following dissipation inequality:
	\begin{align}
		\begin{split}
			\label{eq:finalstep}
			\tilde{V}(x(T),T) - \tilde{V}(x(0),0) + \|\bsmtx \tilde{v} \\ e_I\esmtx\|^2_{2,[0,T]} - \|\bsmtx \tilde{w} \\ \tilde{d}\esmtx\|^2_{2,[0,T]}
			\leq -\tilde{\epsilon} \|\tilde{d}\|^2_{2,[0,T]}
		\end{split}
	\end{align}
	Note that $\tilde{V}(x(0),0)=0$ as $x(0) = 0$ and $\tilde{V}(x(T),T) = x(T)^\top \tilde{P}(T)x(T)$ with boundary condition $\tilde{P}(T) \geq F$ as shown earlier. Apply this boundary condition in inequality \eqref{eq:finalstep} with the definition of $F$ from Equation~\eqref{eq:QSRFRobPerf} to conclude:
	\begin{align}
		\|e_E(T)\|_{2}^2  + \|\bsmtx \tilde{v} \\ e_I\esmtx\|^2_{2,[0,T]} \leq \|\bsmtx \tilde{w} \\ \tilde{d}\esmtx\|^2_{2,[0,T]} -\tilde{\epsilon} \|\tilde{d}\|^2_{2,[0,T]}
	\end{align}
	Divide both sides by $\|\bsmtx \tilde{w} \\ \tilde{d}\esmtx\|^2_{2,[0,T]} < \infty$ and define $\hat{\epsilon}:=\tilde{\epsilon} \|\tilde{d}\|^2_{2,[0,T]}/\|\bsmtx \tilde{w} \\ \tilde{d}\esmtx\|^2_{2,[0,T]}$ to show that $\|N_{scl}\|_{[0,T]} \leq 1 - \hat{\epsilon}$. This proof can be worked backwards to prove $\mathbf{(2 \Rightarrow 1)}$.
\end{proof}

\begin{eremark}
	If the IQC decision variables $M$ and the state-space matrices of $\Psi$ are constant on the given time horizon, then for sufficiently large horizon $T$, the RDE solution for the finite horizon factorization converges to that of the steady state Algebric Riccati Equation (ARE). As a result, the state-space realization of the finite horizon factorization converges to that of an infinite horizon LTI spectral factorization. However, the ARE solution $X$ for infinite horizon spectral factorization is sign indefinite, thus it fails to satisfy the terminal boundary condition on $\tilde{P}$. Thus, it is important to note in the above proof that in order to satisfy the boundary condition on storage function, one must use finite horizon factorization.
\end{eremark}

\nocite{*}
\bibliographystyle{ieeetr} 
\bibliography{References}

\begin{thebibliography}{10}

\bibitem{subrahmanyam2012finite}
M.~B. Subrahmanyam, {\em Finite Horizon {$H_\infty$} and Related Control
  Problems}.
\newblock Springer Science \& Business Media, 2012.

\bibitem{tucker1998continuous}
M.~R. Tucker, {\em Continuous {$H_\infty$} and discrete time-varying finite
  horizon robust control with industrial applications}.
\newblock PhD thesis, University of Leicester, 1998.

\bibitem{marcos09}
A.~Marcos and S.~Bennani, ``{LPV} modeling, analysis and design in space
  systems: Rationale, objectives and limitations,'' in {\em AIAA Guidance,
  Navigation, and Control Conference}, pp.~AIAA 2009--5633, 2009.

\bibitem{biertumpfel2019finite}
F.~Biert{\"u}mpfel, H.~Pifer, and S.~Bennani, ``Finite horizon worst case
  analysis of launch vehicles,'' {\em IFAC-PapersOnLine}, vol.~52, no.~12,
  pp.~31--36, 2019.

\bibitem{man09}
J.~Veenman, H.~K\"{o}ro\u{g}lu, and C.~Scherer, ``Analysis of the controlled
  {NASA HL20} atmospheric re-entry vehicle based on dynamic {IQCs},'' in {\em
  AIAA Guidance, Navigation, and Control Conference}, pp.~AIAA 2009--5637,
  2009.

\bibitem{doyle1987design}
J.~Doyle, K.~Lenz, and A.~Packard, ``Design examples using $\mu$-synthesis:
  Space shuttle lateral axis {FCS} during reentry,'' in {\em Modelling,
  Robustness and Sensitivity Reduction in Control Systems}, Springer, 1987.

\bibitem{balas1992design}
G.~Balas and A.~Packard, ``Design of robust, time-varying controllers for
  missile autopilots,'' in {\em Proceedings of the first IEEE Conference on
  Control Applications}, pp.~104--110, IEEE, 1992.

\bibitem{packard1993linear}
A.~Packard, J.~Doyle, and G.~Balas, ``Linear, multivariable robust control with
  a {$\mu$} perspective,'' {\em Journal of Dynamic Systems, Measurement, and
  Control}, vol.~115, no.~2B, pp.~426--438, 1993.

\bibitem{veenman2014iqc}
J.~Veenman and C.~W. Scherer, ``{IQC}-synthesis with general dynamic
  multipliers,'' {\em International Journal of Robust and Nonlinear Control},
  vol.~24, no.~17, pp.~3027--3056, 2014.

\bibitem{petersen00}
I.~Petersen, V.~Ugrinovskii, and A.~Savkin, {\em Robust Control Design Using
  $H_\infty$ Methods}.
\newblock Springer, 2000.

\bibitem{green2012linear}
M.~Green and D.~J. Limebeer, {\em Linear Robust Control}.
\newblock Courier Corporation, 2012.

\bibitem{tadmor1990worst}
G.~Tadmor, ``Worst-case design in the time domain: The maximum principle and
  the standard {$H_\infty$} problem,'' {\em Mathematics of Control, Signals and
  Systems}, vol.~3, no.~4, pp.~301--324, 1990.

\bibitem{limebeer1992game}
D.~J. Limebeer, B.~D. Anderson, P.~P. Khargonekar, and M.~Green, ``A game
  theoretic approach to {$H_\infty$} control for time-varying systems,'' {\em
  SIAM Journal on Control and Optimization}, vol.~30, no.~2, pp.~262--283,
  1992.

\bibitem{khargonekar1991h_}
P.~P. Khargonekar, K.~M. Nagpal, and K.~R. Poolla, ``{$H_\infty$} control with
  transients,'' {\em SIAM Journal on Control and Optimization}, vol.~29, no.~6,
  pp.~1373--1393, 1991.

\bibitem{ravi1991h}
R.~Ravi, K.~M. Nagpal, and P.~P. Khargonekar, ``{$H_\infty$} control of linear
  time-varying systems: A state-space approach,'' {\em SIAM Journal on Control
  and optimization}, vol.~29, no.~6, pp.~1394--1413, 1991.

\bibitem{uchida1992finite}
K.~{Uchida} and M.~{Fujita}, ``Finite horizon {$H_\infty$} control problems
  with terminal penalties,'' {\em IEEE Transactions on Automatic Control},
  vol.~37, no.~11, pp.~1762--1767, 1992.

\bibitem{lall1995riccati}
S.~Lall and K.~Glover, ``Riccati differential inequalities: suboptimal
  {$H_\infty$} controllers for finite horizon time varying systems,'' in {\em
  Proceedings of 1995 34th IEEE Conference on Decision and Control}, vol.~1,
  pp.~955--956, IEEE, 1995.

\bibitem{yakubovich1967frequency}
V.~Yakubovich, ``Frequency conditions for the absolute stability of control
  systems with several nonlinear or linear nonstationary blocks,'' {\em
  Avtomatika i telemekhanika}, vol.~6, pp.~5--30, 1967.

\bibitem{megretski1997system}
A.~Megretski and A.~Rantzer, ``System analysis via integral quadratic
  constraints,'' {\em IEEE Transactions on Automatic Control}, vol.~42, no.~6,
  pp.~819--830, 1997.

\bibitem{seiler2019finite}
P.~Seiler, R.~M. Moore, C.~Meissen, M.~Arcak, and A.~Packard, ``Finite horizon
  robustness analysis of {LTV} systems using integral quadratic constraints,''
  {\em Automatica}, vol.~100, pp.~135--143, 2019.

\bibitem{wang2016robust}
S.~Wang, H.~Pfifer, and P.~Seiler, ``Robust synthesis for linear parameter
  varying systems using integral quadratic constraints,'' {\em Automatica},
  vol.~68, pp.~111--118, 2016.

\bibitem{veenman2010robust}
J.~Veenman and C.~W. Scherer, ``On robust {LPV} controller synthesis: A dynamic
  integral quadratic constraint based approach,'' in {\em 49th IEEE Conference
  on Decision and Control (CDC)}, pp.~591--596, IEEE, 2010.

\bibitem{o1999robust}
R.~O'Brien~Jr and P.~A. Iglesias, ``Robust controller design for linear,
  time-varying systems,'' {\em European Journal of Control}, vol.~5, no.~2-4,
  pp.~222--241, 1999.

\bibitem{pirie2002robust}
C.~Pirie and G.~E. Dullerud, ``Robust controller synthesis for uncertain
  time-varying systems,'' {\em SIAM Journal on Control and Optimization},
  vol.~40, no.~4, pp.~1312--1331, 2002.

\bibitem{LTVTools}
P.~Seiler, J.~Buch, R.~M. Moore, C.~Meissen, M.~Arcak, and A.~Packard,
  ``{LTVTools (beta), A MATLAB Toolbox for Linear Time-Varying Systems}.''
  \url{https://z.umn.edu/LTVTools}, 2018.

\bibitem{buch2020finte}
J.~Buch and P.~Seiler, ``Finite horizon robust synthesis,'' in {\em Accepted to
  American Control Conference}, IEEE, 2020.

\bibitem{zhou1996robust}
K.~Zhou, J.~C. Doyle, and K.~Glover, {\em Robust and Optimal Control}.
\newblock Prentice Hall, New Jersey, 1996.

\bibitem{dullerud2013course}
G.~E. Dullerud and F.~Paganini, {\em A Course in Robust Control Theory: A
  Convex Approach}, vol.~36.
\newblock Springer Science \& Business Media, 2013.

\bibitem{pfifer2015robustness}
H.~Pfifer and P.~Seiler, ``Robustness analysis with parameter-varying integral
  quadratic constraints,'' in {\em American Control Conference (ACC)},
  pp.~138--143, IEEE, 2015.

\bibitem{balakrishnan2002lyapunov}
V.~Balakrishnan, ``Lyapunov functionals in complex $\mu$ analysis,'' {\em IEEE
  Transactions on Automatic Control}, vol.~47, no.~9, pp.~1466--1479, 2002.

\bibitem{veenman2016robust}
J.~Veenman, C.~W. Scherer, and H.~K{\"o}ro{\u{g}}lu, ``Robust stability and
  performance analysis based on integral quadratic constraints,'' {\em European
  Journal of Control}, vol.~31, pp.~1--32, 2016.

\bibitem{seiler2014stability}
P.~Seiler, ``Stability analysis with dissipation inequalities and integral
  quadratic constraints,'' {\em IEEE Transactions on Automatic Control},
  vol.~60, no.~6, pp.~1704--1709, 2014.

\bibitem{willems1972dissipative1}
J.~C. Willems, ``{Dissipative Dynamical Systems Part I: General Theory},'' {\em
  Archive for rational mechanics and analysis}, vol.~45, no.~5, pp.~321--351,
  1972.

\bibitem{schaft1999}
Schaft, {\em $L_2$-gain and Passivity in Nonlinear Control}.
\newblock Springer-Verlag, 1999.

\bibitem{khalil2002nonlinear}
H.~Khalil, {\em Nonlinear Systems}.
\newblock Prentice Hall, 3rd~ed., 2001.

\bibitem{murray1994mathematical}
R.~Murray, Z.~Li, and S.~Sastry, {\em A Mathematical Introduction to Robot
  Manipulation}.
\newblock CRC Press, 1994.

\bibitem{iannelli2019construction}
A.~Iannelli, P.~Seiler, and A.~Marcos, ``Worst-case disturbances for
  time-varying systems with application to flexible aircraft,'' {\em Journal of
  Guidance, Control, and Dynamics}, vol.~42, no.~6, pp.~1261--1271, 2019.

\bibitem{francis1987course}
B.~Francis, ``A course in {$H_\infty$} control theory,'' {\em Lecture notes in
  Control and Information Sciences}, vol.~88, p.~R5, 1987.

\bibitem{molinari1975nonnegativity}
B.~Molinari, ``Nonnegativity of a quadratic functional,'' {\em SIAM Journal on
  Control}, vol.~13, no.~4, pp.~792--806, 1975.

\bibitem{kuvcera1973review}
V.~Ku{\v{c}}era, ``A review of the matrix riccati equation,'' {\em
  Kybernetika}, vol.~9, no.~1, pp.~42--61, 1973.

\bibitem{lu1999variational}
W.~W. Lu, G.~J. Balas, and E.~B. Lee, ``A variational approach to {$H_\infty$}
  control with transients,'' {\em IEEE Transactions on Automatic Control},
  vol.~44, no.~10, pp.~1875--1879, 1999.

\bibitem{lu1996rigorous}
W.~W. Lu, G.~J. Balas, and E.~B. Lee, ``A rigorous approach to {$H_\infty$}
  control with transients,'' in {\em Proceedings of 35th IEEE Conference on
  Decision and Control}, vol.~2, pp.~2344--2349, IEEE, 1996.

\bibitem{wang2019robust}
H.~Wang, Z.~Shi, and Y.~Zhong, ``Robust finite-horizon optimal control of
  autonomous helicopters in aggressive maneuvering,'' {\em Asian Journal of
  Control}, 2019.

\bibitem{balandin2019finite}
D.~V. Balandin, R.~S. Biryukov, and M.~M. Kogan, ``Finite-horizon
  multi-objective generalized {$H_2$} control with transients,'' {\em
  Automatica}, vol.~106, pp.~27--34, 2019.

\bibitem{balandin2010lmi}
D.~V. Balandin and M.~M. Kogan, ``{LMI}-based {$H_\infty$} optimal control with
  transients,'' {\em International Journal of Control}, vol.~83, no.~8,
  pp.~1664--1673, 2010.

\bibitem{bacsar2008h}
T.~Ba{\c{s}}ar and P.~Bernhard, {\em {$H_\infty$} optimal control and related
  minimax design problems: a dynamic game approach}.
\newblock Springer Science \& Business Media, 2008.

\bibitem{ma01}
D.~Ma and R.~Braatz, ``Worst-case analysis of finite-time control policies,''
  {\em IEEE Transactions on Automatic Control}, vol.~9, no.~5, pp.~766--774,
  2001.

\bibitem{fry2017iqc}
J.~M. Fry, M.~Farhood, and P.~Seiler, ``{IQC}-based robustness analysis of
  discrete-time linear time-varying systems,'' {\em International Journal of
  Robust and Nonlinear Control}, vol.~27, no.~16, pp.~3135--3157, 2017.

\bibitem{Biertumpfel2018WorstCaseGain}
F.~{Biertümpfel} and H.~{Pfifer}, ``Worst case gain computation of linear
  time-varying systems over a finite horizon,'' in {\em 2018 IEEE Conference on
  Control Technology and Applications (CCTA)}, pp.~952--957, 2018.

\bibitem{scherer2018stability}
C.~W. Scherer and J.~Veenman, ``Stability analysis by dynamic dissipation
  inequalities: On merging frequency-domain techniques with time-domain
  conditions,'' {\em Systems \& Control Letters}, vol.~121, pp.~7--15, 2018.

\bibitem{petersen2012robust}
I.~R. Petersen, V.~A. Ugrinovskii, and A.~V. Savkin, {\em Robust Control Design
  Using {$H_\infty$} Methods}.
\newblock Springer Science \& Business Media, 2012.

\bibitem{moore15}
R.~Moore, ``Finite horizon robustness analysis using integral quadratic
  constraints,'' Master's thesis, University of California, Berkeley, 2015.

\bibitem{willems1972dissipative2}
J.~C. Willems, ``{Dissipative Dynamical Systems Part II: Linear Systems with
  Quadratic Supply Rates},'' {\em Archive for rational mechanics and analysis},
  vol.~45, no.~5, pp.~352--393, 1972.

\bibitem{hill1980dissipative}
D.~J. Hill and P.~J. Moylan, ``Dissipative dynamical systems: basic
  input-output and state properties,'' {\em Journal of the Franklin Institute},
  vol.~309, no.~5, pp.~327--357, 1980.

\bibitem{brockett2015finite}
R.~W. Brockett, {\em Finite Dimensional Linear Systems}, vol.~74.
\newblock SIAM, 2015.

\bibitem{pfifer2016less}
H.~Pfifer and P.~Seiler, ``Less conservative robustness analysis of linear
  parameter varying systems using integral quadratic constraints,'' {\em
  International Journal of Robust and Nonlinear Control}, vol.~26, no.~16,
  pp.~3580--3594, 2016.

\bibitem{packard1993complex}
A.~Packard and J.~Doyle, ``The complex structured singular value,'' {\em
  Automatica}, vol.~29, no.~1, pp.~71--109, 1993.

\end{thebibliography}
\end{document}